\documentclass{fundam}
\usepackage{amsmath,amssymb}
\usepackage{graphicx}
\usepackage{url}
\usepackage{hyperref}
\usepackage{algorithm}
\usepackage{algpseudocode}
\usepackage{enumitem}

\newcommand{\Con}{\mathrm{Con}}
\newcommand{\Sub}{\mathrm{Sub}}

\newcommand{\FICSL}{\mathcal{FICSL}}
\newcommand{\FICSLFCP}{\mathcal{FICSL}^{\mathrm{FCP}}}
\newcommand{\Real}{\mathrm{Real}}
\newcommand{\VHFGS}{\mathcal{VH}\text{-}\mathcal{FGSL}}
\newcommand{\ISO}{\mathrm{ISO}}

\newcommand{\rank}{\mathrm{rank}}
\newcommand{\ports}{\mathrm{ports}}

\newcommand{\irank}{\mathrm{irank}}

\newcommand{\BRep}{\mathrm{BRep}}

\algrenewcommand\algorithmicrequire{\textbf{Input:}}
\algrenewcommand\algorithmicensure{\textbf{Output:}}
\algnewcommand\algorithmicoracle{\textbf{Oracle:}}
\algnewcommand\Oracle{\item[\algorithmicoracle]}
\algrenewcommand\algorithmiccomment[1]{// #1}

\begin{document}

\title{Distributional Learning of Graph Languages Generated by Fixed-Interface Clause Systems}

\author{
Takayoshi Shoudai\\
Department of Computer Science and Engineering,
Fukuoka Institute of Technology\\
Fukuoka 811-0295, Japan\\
\texttt{shodai@fit.ac.jp}
\and
Satoshi Matsumoto\\
Department of Mathematical Sciences,
Tokai University\\
Hiratsuka 259-1292, Japan\\
\texttt{matsumoto@tokai.ac.jp}
\and
Yusuke Suzuki\\
Graduate School of Information Sciences,
Hiroshima City University\\
Hiroshima 731-3194, Japan\\
\texttt{y-suzuki@hiroshima-cu.ac.jp}
\and
Tomoyuki Uchida\\
Graduate School of Information Sciences,
Hiroshima City University\\
Hiroshima 731-3194, Japan\\
\texttt{uchida@hiroshima-cu.ac.jp}
}

\runninghead{T. Shoudai, S. Matsumoto, Y. Suzuki, T. Uchida}{Distributional Learning of Fixed-Interface Graph Languages}

\maketitle

\begin{abstract}
Distributional learning provides a useful framework for studying the learnability of structured languages from positive data.
In this paper, we extend this framework to graph languages generated by fixed-interface clause systems.

We give an explicit formulation of fixed-interface clause systems (FICSs) and study the corresponding learning problem under positive presentations and membership queries.
Within this setting, we consider a bounded class of graph languages satisfying the finite context property (FCP) under a bounded-degree assumption.
The bounds are expressed by the degree bound $\Delta$ together with five structural parameters $m,s,t,w$, and $d$: these control the degree of generated graphs, the number of clauses, the number of variable hyperedges, the number of body atoms, interface ranks, and local head-frame complexity, respectively.

The learning algorithm constructs hypotheses from ordered boundary representations induced by the observed positive examples.
These representations make explicit the interface information needed to compare contexts and to test candidate clauses by membership queries.
We prove the correctness of the algorithm step by step: target contexts eventually appear in the observed sample, target clauses are reconstructed over the corresponding predicate representatives, and spurious non-fact clauses are eventually excluded by membership queries.
As a consequence, for every fixed parameter tuple, the target language is identifiable in the limit from positive data and membership queries.

We also analyze the complexity of the hypothesis update process.
Using the bounded-degree assumption, the fixed interface-rank bound, and the finite head-frame condition, we show that only polynomially many ordered boundary representations, predicate symbols, clause candidates, and membership queries are needed at each stage.
Hence the learner has polynomial-time update on $\FICSLFCP_{\Delta}(m,s,t,w,d)$, the class of graph languages generated by such bounded FICSs with the finite context property.

Overall, the paper provides a parameterized reformulation of distributional learning for interface-based graph languages in a fixed-interface setting, making the roles of boundary information, clause structure, and structural parameters explicit.
\end{abstract}


\section{Introduction}\label{sec:introduction}
Distributional learning has been one of the central approaches in algorithmic learning theory for establishing learnability from positive data.
It was initiated by Clark and Eyraud \cite{clark2007} for subclasses of context-free grammars and was subsequently extended to broader grammar classes and to more general structured objects \cite{kasprzik2011,yoshinaka2012,yoshinaka2015}.
More recently, this line of research has continued to develop on the string side, where richer non-terminal representations have been used to extend distributional learning from positive data and membership queries to larger classes of context-free languages \cite{kanazawa2023}.
A common theme throughout these studies is that learnability becomes accessible once the target objects admit an appropriate decomposition into contexts and substructures.

The present paper is concerned with graph languages.
As graph-based counterparts of grammatical formalisms for strings and trees, graph grammars have been studied in several variants.
Among graph grammar formalisms, hyperedge replacement grammars (HRGs) provide a standard context-free mechanism for graph generation: a nonterminal hyperedge is replaced by a graph fragment through designated attachment vertices \cite{Habel1992,drewes97}.
Formal graph systems (FGSs), introduced by Uchida et al. \cite{uchida1995parallel}, provide a logic-programming-based formalism in which graph patterns serve as structured objects.
Regular FGSs can be regarded as a logic-programming presentation of the interface-based generation mechanism underlying HRGs, in analogy with clause-based formalisms for strings such as elementary formal systems~\cite{Smullyan1961EFS,Arikawa1970}.
Regular and restricted variants of FGSs have been used as a basis for studying learnability properties of graph languages \cite{Hara2014,shoudai2016ilp,shoudai2023pac}.
The fixed-interface clause systems used in the present paper build on this viewpoint.
They retain the HRG-inspired interface replacement mechanism, but make boundary vertices and port orders explicit and allow repeated variable labels in clauses.
The last feature is important because repeated occurrences of the same variable label force the corresponding positions to be replaced by isomorphic copies of the same graph with interface, which is not part of the standard HRG replacement mechanism.
Thus the formalism used here is not simply a restatement of regular FGSs; rather, it is a fixed-interface, logic-programming-style extension of the HRG/regular-FGS generation mechanism, tailored to the distributional-learning argument developed below.

Several learning-theoretic results for graph-based formalisms are already known.
Okada et al. \cite{okada2007learning} established exact learning results for certain graph pattern classes in the minimally adequate teacher model.
Hara and Shoudai \cite{Hara2014} studied polynomial-time MAT learning of c-deterministic regular FGSs.
Shoudai et al. \cite{shoudai2016ilp} investigated regular FGS languages of bounded degree under positive data and membership queries in a distributional-learning setting.
More recently, Shoudai et al. \cite{shoudai2023pac} studied polynomial-time PAC learnability for bounded classes of variable-hereditary FGS languages under bounded degree and bounded treewidth assumptions.
These results indicate that graph languages admit a substantial learning theory, but they also suggest that the roles of boundary information, clause structure, and structural parameters deserve a more explicit treatment than in the earlier formulations.

The aim of this paper is to formulate and analyze a distributional-learning framework for graph languages generated by fixed-interface clause systems.
The point is not merely to restate an earlier result in different notation.
Rather, by developing this fixed-interface formulation, we make explicit the boundary structure that is already implicit in fragment-based decompositions.
Although the target language consists of plain graphs, the learning procedure works with graphs with interfaces obtained by cutting observed graphs along designated boundary vertices.
This interface structure is essential both for the definition of context candidates and for the parameterized complexity analysis of hypothesis updates.
Within this framework, we consider a bounded class of graph languages satisfying the finite context property under a bounded-degree assumption.
The bounds are given explicitly by the degree bound $\Delta$ and the structural parameter tuple $(m,s,t,w,d)$: the former controls the generated graph class through the maximum degree, while the latter controls, among other things, the number of clauses, the number of variable hyperedges, the number of body atoms, the interface ranks, and the finite family of local head-frame types.
This parameterized formulation makes it possible to isolate which structural restrictions are responsible for learnability and which are responsible for polynomial-time update.
In particular, compared with our earlier formulation~\cite{shoudai2016ilp}, the present formulation replaces the role of a fixed Chomsky-normal-form-style bound on rule shapes by explicit parameter bounds that can be traced throughout the correctness and complexity arguments.

For this class, we present an oracle-guided learning algorithm based on ordered boundary representations induced by the observed positive examples.
The algorithm constructs hypotheses by introducing predicate representatives from these boundary representations and by admitting bounded fixed-interface clause candidates through membership-query tests.
We prove, step by step, that target contexts eventually occur in the observed sample, that target clauses are reconstructed over the corresponding predicate representatives, and that spurious non-fact clauses are eventually excluded by membership queries.
We also analyze the update complexity and prove that the learner has polynomial-time update.

This paper is an extended journal version of our earlier conference paper~\cite{shoudai2016ilp}, in which all proofs were omitted.
Beyond providing complete proofs, the present paper makes three further contributions.
First, we reformulate the framework using fixed-interface clause systems and ordered boundary representations, making the boundary structure explicit rather than leaving it implicit in fragment-based decompositions.
Second, we make explicit that the bounded treewidth assumption used in \cite{shoudai2016ilp} is not needed for the present distributional-learning argument.
The result holds under the bounded-degree assumption together with the fixed structural parameters.
Bounded treewidth was a natural restriction in our earlier work \cite{shoudai2017} on graph-pattern matching.
In that setting, tree decompositions were used to obtain polynomial-time algorithms for partial $k$-tree patterns.
In contrast, the present distributional-learning framework does not use tree decompositions in the construction of context candidates or in the membership-query filtering argument.
This observation clarifies why the bounded treewidth assumption imposed in the earlier formulation~\cite{shoudai2016ilp} is not essential for the present learning argument.
Third, we introduce an explicit parameter tuple controlling the generated graph class and the clause-system structure, allowing the restrictions responsible for learnability and polynomial-time update to be traced throughout the correctness and complexity arguments.
At the same time, since the construction is based on context decomposition and membership-query filtering, the paper can also be read as a graph-structured generalization of the corresponding string-based distributional-learning framework \cite{kanazawa2023}.

The remainder of the paper is organized as follows.
Section~\ref{sec:preliminaries} introduces graphs with interfaces, graph patterns, and fixed-interface clause systems.
Section~\ref{sec:learning-model} defines the learning model.
Section~\ref{sec:main-results} presents the learning algorithm, proves its correctness, and establishes the parameterized polynomial-time update bound.

\section{Preliminaries}\label{sec:preliminaries}
Throughout the paper, all graphs are finite, undirected, and labeled, unless explicitly stated otherwise.
Thus, a graph is a tuple $G=(V,E,\phi,\psi)$, where $V$ is a finite set of vertices, $E\subseteq\{\{u,v\}\mid u,v\in V,\ u\neq v\}$ is a finite set of undirected edges, and $\phi$ and $\psi$ are vertex- and edge-labeling functions, respectively.
For a graph $G$, we write $V(G)$ and $E(G)$ for its vertex set and edge set, respectively.
For a finite set $X$, we denote by $|X|$ its cardinality.
For a finite ordered tuple $\iota$, we denote by $|\iota|$ its length.

\subsection{Graphs with interfaces and ordered boundary representations}

We first recall the notion of a graph with interface and introduce ordered boundary representations.
Interfaces specify the vertices through which a graph fragment can be composed with another fragment.
Ordered boundary representations will later be used to enumerate, from positive examples, the finitely many context candidates needed by the learning algorithm.

\begin{definition}[Graph with interface]
Let $\Sigma_V$ and $\Sigma_E$ be finite alphabets of vertex labels and edge labels, respectively.
A graph with interface is a tuple $P=(V,E,\phi,\psi,\iota)$, where
\begin{enumerate}[label=\textup{(\arabic*)}]
\item $V$ is a finite set of vertices,
\item $E \subseteq \bigl\{\{u,v\}\mid u,v\in V,\ u\neq v\bigr\}$ is a finite set of edges,
\item $\phi:V \to \Sigma_V$ is a vertex-labeling function,
\item $\psi:E \to \Sigma_E$ is an edge-labeling function, and
\item $\iota=(v_1,\ldots,v_r)$ is an ordered list of distinct vertices in $V$.
\end{enumerate}
The vertices appearing in $\iota$ are called the \emph{interface vertices} of $P$, and the integer $r$ is called the \emph{interface rank} of $P$.
\end{definition}

\begin{remark}
Every graph $(V,E,\phi,\psi)$ is identified with the graph with the empty interface obtained by regarding it as $(V,E,\phi,\psi,())$.
\end{remark}

\begin{definition}[Subgraph with interface]\label{def:subgraph-with-interface}
Let $G=(V,E,\phi,\psi)$ be a graph.
A \emph{subgraph with interface} of $G$ is a graph with interface $K=(V_K,E_K,\phi_K,\psi_K,\iota_K)$ such that $V_K\subseteq V,\: E_K\subseteq E$, $\phi_K=\phi|_{V_K},\: \psi_K=\psi|_{E_K}$, and $\iota_K=(b_1,\ldots,b_r)$ is an ordered list of distinct vertices of $V_K$.
The vertices $b_1,\ldots,b_r$ are called the \emph{boundary vertices} of $K$, and $r$ is called the \emph{boundary rank} of $K$.
\end{definition}

\begin{definition}[Ordered boundary representation]\label{def:ordered-boundary-representation}
Let $G=(V,E,\phi,\psi)$ be a graph, and let $K=(V_K,E_K,\phi_K,\psi_K,\iota_K)$ be a subgraph with interface of $G$, where $\iota_K=(b_1,\ldots,b_r)$.
Let $B_K=\{b_1,\ldots,b_r\}$.
The \emph{boundary edge set} of $K$ is $E_B(K)=\{e\in E_K\mid e\cap B_K\neq\emptyset\}$.
The \emph{ordered boundary representation} of $K$ is the pair $(\beta,E_B(K))$, where $\beta=(b_1,\ldots,b_r)$.
The order of the vertices in $\beta$ is part of the representation.
Thus, the same boundary vertex set with different orders gives different ordered boundary representations.
\end{definition}

\begin{definition}[Valid ordered boundary specification]\label{def:valid-ordered-boundary-specification}
Let $G=(V,E,\phi,\psi)$ be a graph.
An \emph{ordered boundary specification} in $G$ is a pair $(\beta,E_B)$, where $\beta=(b_1,\ldots,b_r)$ is an ordered list of distinct vertices of $G$, and $E_B\subseteq E$.
Let $B=\{b_1,\ldots,b_r\}$.
We say that $(\beta,E_B)$ is \emph{valid} in $G$ if every edge in $E_B$ is incident with at least one vertex in $B$, that is, $e\cap B\neq\emptyset\:\text{for every } e\in E_B$.

For a valid ordered boundary specification $(\beta,E_B)$, define $S_G(\beta,E_B)=\{v\in V\setminus B \mid \{b,v\}\in E_B \text{ for some } b\in B\}$.
Let $I_G(\beta,E_B)$ be the set of all vertices of $G-B$ that are reachable in $G-B$ from some vertex in $S_G(\beta,E_B)$.
Here $G-B$ denotes the graph obtained from $G$ by deleting all vertices in $B$ and all edges incident with them.

The subgraph with interface determined by $(\beta,E_B)$ is $K_G(\beta,E_B)=(V_K,E_K,\phi_K,\psi_K,\beta)$, where $V_K=B\cup I_G(\beta,E_B)$ and $E_K=E_B \cup\bigl\{\{u,v\}\in E\mid u,v\in I_G(\beta,E_B)\bigr\}$.
The label functions are the restrictions $\phi_K=\phi|_{V_K}$, $\psi_K=\psi|_{E_K}$.
\end{definition}

\begin{remark}
In this paper, the term ``boundary'' refers to an ordered interface used for composition, and the boundary vertices of a specification in $G$ are not required to form a separator of $G$.
An arbitrary pair $(\beta,E_B)$ need not be a valid ordered boundary specification.
In the construction below, such pairs are enumerated as candidates, and only valid ones are used to construct subgraphs with interfaces.

In the learning construction below, contexts are treated as fragments obtained by cutting positive examples along ordered interfaces.
Such fragments are represented by valid ordered boundary specifications in the sense of Definition~\ref{def:valid-ordered-boundary-specification}.
Thus the context candidates considered by the learner are not arbitrary subgraphs with interfaces; they are the subgraphs with interfaces obtained from ordered boundary vertices and chosen boundary edge sets by the construction of Definition~\ref{def:valid-ordered-boundary-specification}.
\end{remark}

\begin{lemma}\label{lem:boundary-specification-unique}
Let $G=(V,E,\phi,\psi)$ be a graph.
Every valid ordered boundary specification $(\beta,E_B)$ in $G$ determines the unique subgraph with interface $K_G(\beta,E_B)$.
\end{lemma}

\begin{proof}
Let $\beta=(b_1,\ldots,b_r)$ and put $B=\{b_1,\ldots,b_r\}$.
Since $(\beta,E_B)$ is valid, every edge in $E_B$ is an edge of $G$ incident with at least one vertex in $B$.
The set $S_G(\beta,E_B)$ is determined uniquely by $G$, $B$, and $E_B$.
The graph $G-B$ is also uniquely determined by $G$ and $B$.
Hence the set $I_G(\beta,E_B)$ of vertices reachable in $G-B$ from $S_G(\beta,E_B)$ is uniquely
determined.

Therefore the vertex set $V_K=B\cup I_G(\beta,E_B)$
is uniquely determined.
The edge set $E_K=E_B\cup\bigl\{\{u,v\}\in E\mid u,v\in I_G(\beta,E_B)\bigr\}$ is also uniquely determined.
The label functions $\phi_K=\phi|_{V_K}$,$\psi_K=\psi|_{E_K}$ are uniquely determined by the label functions of $G$, and the interface is exactly the ordered tuple $\beta$.

Thus $K_G(\beta,E_B)=(V_K,E_K,\phi_K,\psi_K,\beta)$ is uniquely determined.
\end{proof}

\begin{lemma}\label{lem:boundary-specification-decision}
Let $G=(V,E,\phi,\psi)$ be a graph.
Given an ordered tuple $\beta=(b_1,\ldots,b_r)$ of vertices and a set $E_B$, it can be decided in time $O(|V|+|E|)$ whether $(\beta,E_B)$ is a valid ordered boundary specification in $G$.
If it is valid, the corresponding subgraph with interface $K_G(\beta,E_B)$ can be constructed in time $O(|V|+|E|)$.
\end{lemma}

\begin{proof}
First, we check whether $b_1,\ldots,b_r$ are vertices of $G$ and are
pairwise distinct. This can be done in $O(r)$ time using marking on the
vertex set. Next, let $B=\{b_1,\ldots,b_r\}$.
For every edge $e \in E_B$, we check whether $e \in E$ and whether $e \cap B \neq \emptyset$.
Using a standard adjacency representation, this takes $O(|E_B|)$ time after the graph has been stored in a form that allows us to test whether a given unordered pair of vertices is an edge.
Since $E_B\subseteq E$, we have $|E_B|\leq |E|$.
Thus this validity check is bounded by $O(|V|+|E|)$.
If the check fails, then $(\beta,E_B)$ is not valid.

Suppose that the check succeeds.
We compute $S_G(\beta,E_B)=\{v\in V\setminus B\mid\{b,v\}\in E_B \text{ for some } b\in B\}$.
This is done by scanning the edges in $E_B$, and therefore takes $O(|E_B|)$ time.
We then perform a breadth-first search or depth-first search in the graph $G-B$, starting from all vertices in $S_G(\beta,E_B)$.
This computes $I_G(\beta,E_B)$ in time $O(|V|+|E|)$.
Finally, we construct $V_K=B\cup I_G(\beta,E_B)$ and $E_K=E_B\cup\bigl\{\{u,v\}\in E\mid u,v\in I_G(\beta,E_B)\bigr\}$ by scanning the edge set $E$ once.
This also takes $O(|V|+|E|)$ time.
The label functions are obtained by restricting $\phi$ and $\psi$, which adds only linear time.

Therefore validity can be decided, and the corresponding subgraph with interface can be constructed when valid, in time $O(|V|+|E|)$.
\end{proof}

\begin{lemma}\label{lem:number-of-boundary-specifications}
Let $G=(V,E,\phi,\psi)$ be a graph of maximum degree at most $\Delta$, and put $N=|V|$.
For every fixed $r\geq 0$, the number of ordered boundary specifications $(\beta,E_B)$ of boundary rank $r$ is at most $N^r 2^{r\Delta}$.
\end{lemma}

\begin{proof}
First, we count the number of possible ordered boundary tuples $\beta=(b_1,\ldots,b_r)$ of distinct vertices of $G$.
This number is $N(N-1)\cdots(N-r+1)\leq N^r$.
Fix such a tuple $\beta$, and let $B=\{b_1,\ldots,b_r\}$.
Since $G$ has maximum degree at most $\Delta$, the number of edges of $G$ incident with at least one vertex of $B$ is at most $r\Delta$.
The boundary edge set $E_B$ is chosen as a subset of these incident edges.
Therefore the number of possible choices of $E_B$ is at most $2^{r\Delta}$.
Multiplying the number of choices of $\beta$ and $E_B$, we obtain the upper bound $N^r 2^{r\Delta}$.
This proves the lemma.
\end{proof}

\begin{lemma}\label{lem:boundary-specification-enumeration-time}
Let $G=(V,E,\phi,\psi)$ be a graph of maximum degree at most $\Delta$, and put $N=|V|$.
For a fixed boundary-rank bound $w$, all valid ordered boundary specifications of rank at most $w$, together with their corresponding subgraphs with interfaces, can be enumerated in time $O\bigl(\Delta\,2^{w\Delta} N^{w+1}\bigr)$.
In particular, if $w$ and $\Delta$ are fixed constants, this time is polynomial in $N$.
\end{lemma}

\begin{proof}
Fix $r\leq w$. By Lemma~\ref{lem:number-of-boundary-specifications}, the number of candidate ordered boundary specifications of rank $r$ is at most $N^r 2^{r\Delta}$.
For each candidate $(\beta,E_B)$, Lemma~\ref{lem:boundary-specification-decision} shows that validity can be decided and, if valid, the corresponding subgraph with interface can be constructed in time $O(|V|+|E|)$.
Since $G$ has maximum degree at most $\Delta$, we have $|E|=O(\Delta N)$.
Hence the time per candidate is $O(\Delta N)$.

Therefore, for rank $r$, the total time is bounded by $O\bigl(N^r 2^{r\Delta}\cdot \Delta N\bigr)$.
Summing over $r=0,1,\ldots,w$, we obtain
\[
\sum_{r=0}^{w}
O\bigl(N^r 2^{r\Delta}\cdot \Delta N\bigr)
=
O\bigl(\Delta\,2^{w\Delta} N^{w+1}\bigr),
\]
because $w$ is fixed and the largest term occurs at $r=w$ up to a constant factor depending only on $w$ and $\Delta$.
Thus all valid ordered boundary specifications of rank at most $w$, and their corresponding subgraphs with interfaces, can be enumerated within the claimed time bound.
\end{proof}

\begin{corollary}\label{cor:sample-boundary-enumeration-time}
Let $D_n=\{G_1,\ldots,G_n\}$ be a finite set of positive examples.
Let
\[
S_n=\sum_{i=1}^{n}|V(G_i)|.
\]
Assume that every $G_i$ has maximum degree at most $\Delta$.
For fixed $w$ and $\Delta$, all valid ordered boundary specifications of rank at most $w$ arising from the graphs in $D_n$, together with their corresponding subgraphs with interfaces, can be enumerated in time polynomial in $S_n$.
More precisely, the total time is bounded by
\[
O\left(
\Delta\,2^{w\Delta}
\sum_{i=1}^{n}|V(G_i)|^{w+1}
\right)
\leq
O\left(
\Delta\,2^{w\Delta} S_n^{w+1}
\right).
\]
\end{corollary}

\begin{proof}
Apply Lemma~\ref{lem:boundary-specification-enumeration-time} to each graph
$G_i$. If $N_i=|V(G_i)|$, then the time required for $G_i$ is $O\bigl(\Delta\,2^{w\Delta} N_i^{w+1}\bigr)$.
Summing over $i=1,\ldots,n$, the total time is $O\left(\Delta\,2^{w\Delta}\sum_{i=1}^{n} N_i^{w+1}\right)$.
Since
\[
\sum_{i=1}^{n} N_i^{w+1}
\leq
\left(\sum_{i=1}^{n} N_i\right)^{w+1}
=
S_n^{w+1},
\]
we obtain $O\left(\Delta\,2^{w\Delta} S_n^{w+1}\right)$.
For fixed $w$ and $\Delta$, this is polynomial in $S_n$.
\end{proof}

\subsection{Graph patterns with variable hyperedges}

\begin{definition}[Graph pattern]
Let $\Sigma_V$ and $\Sigma_E$ be finite alphabets of vertex labels and edge labels, respectively, and let $X$ be an infinite set of variable labels equipped with a rank function
\[
\rank:X \to \mathbb{N}.
\]
A \emph{graph pattern with variable hyperedges} is a tuple 
\[
G=(V,E,H,\phi,\psi,\lambda,\ports,\iota),
\]
where
\begin{enumerate}[label=\textup{(\arabic*)}]
\item $V$ is a finite set of vertices,
\item $E \subseteq \bigl\{\{u,v\}\mid u,v\in V,\ u\neq v\bigr\}$ is a finite set of ordinary edges,
\item $H$ is a finite set of hyperedges,
\item $\phi:V \to \Sigma_V$ is a vertex-labeling function,
\item $\psi:E \to \Sigma_E$ is an edge-labeling function,
\item $\lambda:H \to X$ is a variable-labeling function,
\item for each $h \in H$, $\ports(h)$ is an ordered list of distinct vertices of
      length $\rank(\lambda(h))$, and
\item $\iota$ is an ordered list of distinct vertices in $V$.
\end{enumerate}
The ordered list $\iota$ is called the \emph{interface} of $G$.
A graph pattern with variable hyperedges is said to be \emph{ground} if $H=\emptyset$.
In this case, it is identified with the graph with interface $(V,E,\phi,\psi,\iota)$.
\end{definition}

\begin{example}
We give an example of a graph pattern with variable hyperedges in Figure~\ref{fig:gpattern}.
A variable hyperedge is drawn as a box with lines to its ports.
The order of the ports is indicated by the digits attached to these lines.
The vertices $v_1,v_2,v_3$ form the interface of the graph pattern.
\end{example}

\begin{figure}[tb]
 \begin{center}
  \includegraphics[scale=0.48]{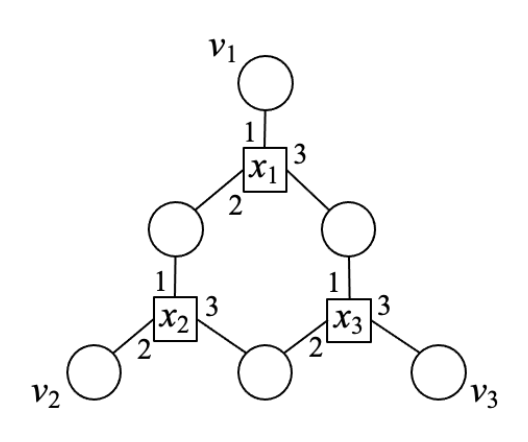}
  \caption{
  A graph pattern $G=(V,E,H,\varphi,\psi,\lambda,ports,\iota)$ with variable hyperedges.
  In this example, $E=\emptyset$, $H=\{h_1,h_2,h_3\}$, $\lambda(h_i)=x_i$ for $i=1,2,3$, and $\iota=(v_1,v_2,v_3)$.
  The port order of each hyperedge is indicated by the digits attached to the incident lines; in particular, each variable $x_i$ has rank $3$.
  }
  \label{fig:gpattern}
 \end{center}
\end{figure}

\begin{definition}[Star graph pattern]
A graph pattern with variable hyperedges $G=(V,E,H,\phi,\psi,\lambda,\ports,\iota)$ is called a \emph{star graph pattern} if
\begin{enumerate}[label=\textup{(\arabic*)}]
\item $E=\emptyset$,
\item $H=\{h\}$ for some hyperedge $h$,
\item the interface $\iota$ coincides with\/ $\ports(h)$, and
\item $V$ is exactly the set of vertices occurring in $\ports(h)$.
\end{enumerate}
\end{definition}

\begin{definition}[Isomorphism of graph patterns]
Let $G$ and $G'$ be graph patterns with variable hyperedges, where $G=(V,E,H,\phi,\psi,\lambda,\ports,\iota)$ and $G'=(V',E',H',\phi',\psi',\lambda',\ports',\iota')$.
An isomorphism from $G$ to $G'$ is a pair $(\pi,\tau)$ consisting of bijections $\pi:V \to V'$ and $\tau:H \to H'$ such that
\begin{enumerate}[label=\textup{(\arabic*)}]
\item $\{u,v\}\in E$ if and only if $\{\pi(u),\pi(v)\}\in E'$,
\item $\phi(v)=\phi'(\pi(v))$ for every $v\in V$,
\item $\psi(\{u,v\})=\psi'(\{\pi(u),\pi(v)\})$ for every $\{u,v\}\in E$,
\item $\lambda(h)=\lambda'(\tau(h))$ for every $h\in H$,
\item if $\ports(h)=(v_1,\ldots,v_r)$, then $\ports'(\tau(h))=(\pi(v_1),\ldots,\pi(v_r))$,
\item if $\iota=(u_1,\ldots,u_m)$ and $\iota'=(u'_1,\ldots,u'_m)$, then $\pi(u_i)=u'_i \quad (i=1,\ldots,m)$.
\end{enumerate}
If there exists an isomorphism from $G$ to $G'$, then we write $G \cong G'.$
\end{definition}

\begin{definition}[Binding and substitution]\label{def:binding-substitution}
Let $K=(V_K,E_K,\phi_K,\psi_K,\iota_K)$ be a ground graph pattern with interface, and let $x \in X$ with $\rank(x)=|\iota_K|$.
A \emph{binding} for $x$ is an expression $x := K$.

Let $G=(V,E,H,\phi,\psi,\lambda,\ports,\iota)$ be a graph pattern with variable hyperedges, and let $h \in H$ satisfy $\lambda(h)=x$.
Replacing the hyperedge occurrence $h$ with $K$ means removing $h$, taking a fresh copy of $K$, and identifying the $i$-th vertex of $\ports(h)$ with the $i$-th interface vertex of the copied $K$ for each $i=1,\ldots,\rank(x)$.
If $G$ contains several hyperedges labeled with the same variable $x$, then a binding $x:=K$ is applied simultaneously to all such hyperedges.
In other words, all occurrences of the same variable label are replaced by isomorphic copies of the same ground graph pattern $K$.

A \emph{substitution} is a finite set $\theta=\{x_1:=K_1,\ldots,x_m:=K_m\}$ such that the variables $x_1,\ldots,x_m$ are pairwise distinct and each $K_i$ is a ground graph pattern with interface satisfying $\rank(x_i)=|\iota_{K_i}|$.
For a graph pattern $G$, we write $G\theta$ for the graph pattern obtained by applying all bindings in $\theta$ simultaneously.
\end{definition}
\subsection{Clause Systems with Fixed Interfaces}
\label{subsec:fixed-interface-clause-systems}

We now introduce the clause-system formalism used in the learning algorithm.
The graph patterns defined in the previous subsection provide the local building blocks of clauses. Variable hyperedges indicate positions at which graphs with interfaces may later be substituted.
The role of fixed-interface clauses is to make this substitution structure explicit. In each clause, variable hyperedges occurring in the head pattern are linked to star graph patterns in the body.
Thus the body atoms specify the subproblems associated with the interfaces of the variable hyperedges.
This formulation is tailored to the later learning construction, where contexts are represented by ordered boundary representations and candidate clauses are tested by membership queries.

\begin{definition}[Atom and clause]
Let $\Pi$ be a finite set of predicate symbols, and let
\[
\irank:\Pi \to \mathbb{N}
\]
be an interface-rank function.
An \emph{atom} is an expression of the form $p(G)$, where $p \in \Pi$ and $G$ is a graph pattern with variable hyperedges such that $|\iota|=\irank(p)$, where $\iota$ is the interface of $G$.
For an atom $p(G)$ and a substitution $\theta$, we define $p(G)\theta := p(G\theta)$.

A \emph{clause} is an expression of the form
\[
A \leftarrow B_1,\ldots,B_m,
\]
where $m \ge 0$ and $A,B_1,\ldots,B_m$ are atoms.
The atom $A$ is called the \emph{head} of the clause, and the sequence $B_1,\ldots,B_m$ is called the \emph{body} of the clause.
If $m=0$, then the clause is called a \emph{fact}.
A finite set of clauses is called a \emph{graph pattern clause system}.
\end{definition}

\begin{definition}[Fixed-interface clause]\label{def:fixed-interface-clause}
Let $A\leftarrow B_1,\ldots,B_m$ be a clause, and write $A=p(G_A)$, where
\[
G_A=(V_A,E_A,H_A,\phi_A,\psi_A,\lambda_A,\ports_A,\iota_A)
\]
is the head pattern of $A$.
The clause $A\leftarrow B_1,\ldots,B_m$ is called a \emph{fixed-interface clause} if the following conditions hold:
\begin{enumerate}[label=\textup{(\arabic*)}]
\item $G_A$ is a graph pattern with variable hyperedges,
\item every body atom has a star graph pattern,
\item for each variable hyperedge occurrence $h\in H_A$, there exists exactly one body atom whose graph pattern is a star graph pattern $G_h=(V_h,E_h,H_h,\phi_h,\psi_h,\lambda_h,\ports_h,\iota_h)$ with $H_h=\{h'\}$ satisfying $\lambda_h(h')=\lambda_A(h)$ and $|\iota_h|=\rank(\lambda_A(h))$.
Conversely, every body atom arises in this way from exactly one variable hyperedge occurrence in the head pattern.
\end{enumerate}
\end{definition}

\begin{definition}[Fixed-interface clause system (FICS)]
A graph pattern clause system is called a \emph{fixed-interface clause system} if every clause in it is a fixed-interface clause.
\end{definition}

\begin{remark}
The fixed-interface clause systems considered in this paper make explicit, in graph-with-interface notation, the structural ingredients already present in regular FGSs~\cite{uchida1995parallel,shoudai2016ilp}: designated interface vertices, star-shaped body patterns, and the correspondence between variable hyperedges in the head pattern and body atoms through variable labels.
In particular, every regular FGS rule with unary predicate symbols can be viewed as a fixed-interface clause of the kind considered here, after making its interface vertices and port order explicit.
Conversely, fixed-interface clauses satisfying the syntactic restrictions of regular FGSs give regular-FGS-style rules~\cite{uchida1995parallel,shoudai2016ilp}.
The present paper, however, uses the fixed-interface formulation as the primary formalism, since the learning argument depends only on the interface structure, bounded clause shape, and membership-query filtering, rather than on the full regular-FGS presentation.
\end{remark}

\begin{example}
Figure~\ref{fig:fics-repeated-variable} shows a fixed-interface clause system in which the same variable label may occur in several variable hyperedges.
In such a clause, all occurrences of the same variable label are replaced by isomorphic copies of the same graph with interface under a substitution.
\begin{figure}[t]
\centering
\includegraphics[scale=0.48]{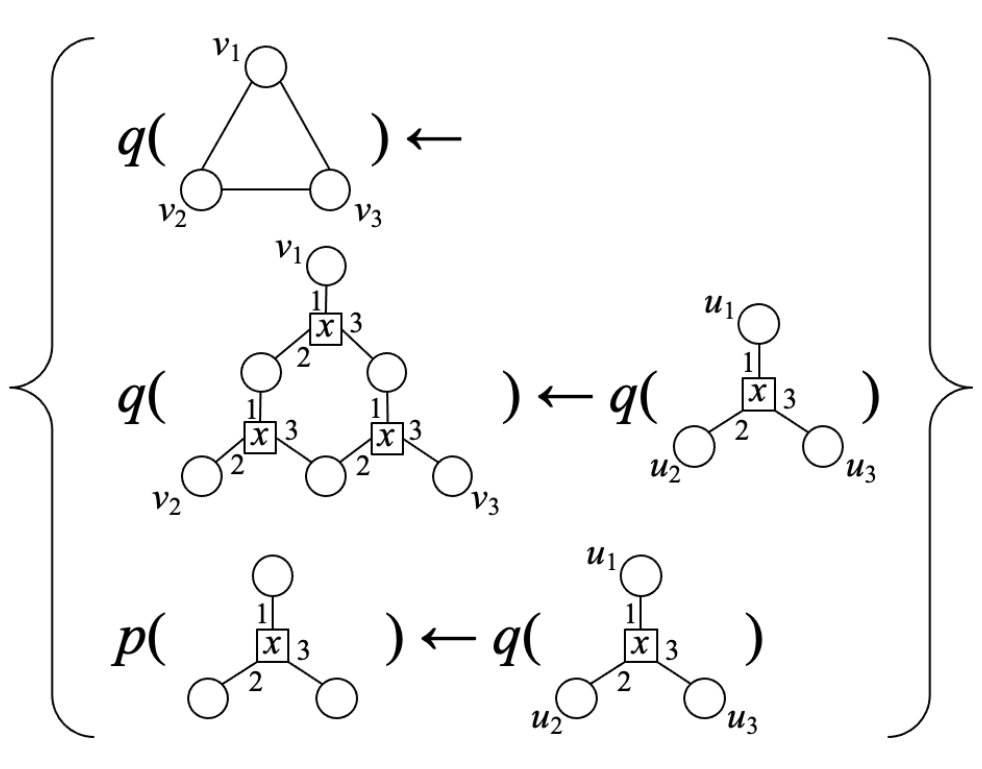}
\caption{A fixed-interface clause system with repeated variable labels. Vertex names are local to each displayed clause; identical variable labels are replaced by isomorphic copies of the same graph with interface.}
\label{fig:fics-repeated-variable}
\end{figure}
\end{example}

\begin{definition}[Derivation]
Let $\Gamma=\{\gamma_1,\ldots,\gamma_t\}$ be a graph pattern clause system.
A \emph{goal} is an expression of the form
\[
\leftarrow B_1,\ldots,B_m,
\]
where $m \ge 0$ and $B_1,\ldots,B_m$ are atoms.
The case $m=0$ is called the \emph{empty goal}.

Let $\mathfrak{g}=\ \leftarrow A_1,\ldots,A_k$ be a nonempty goal, and let $\gamma=A \leftarrow B_1,\ldots,B_q$ be a clause in $\Gamma$.
Suppose that $A_j=p(G)$ and $A=p(H)$ for some $j \in \{1,\ldots,k\}$, and that there exists a substitution $\theta$ such that $H\theta \cong G$.
Then we say that the goal
\[
\mathfrak{g}'=\ \leftarrow A_1,\ldots,A_{j-1},B_1\theta,\ldots,B_q\theta,A_{j+1},\ldots,A_k
\]
is obtained from $\mathfrak{g}$ by one \emph{derivation step} using $\gamma$ at $A_j$.

A \emph{derivation} from a goal $\mathfrak{g}_0$ is a finite or infinite sequence $\mathfrak{g}_0,\mathfrak{g}_1,\mathfrak{g}_2,\ldots$ such that each $\mathfrak{g}_{i+1}$ is obtained from $\mathfrak{g}_i$ by one derivation step.
A finite derivation ending with the empty goal is called a \emph{refutation}.
\end{definition}

\begin{definition}[Generated language]
Let $\Gamma$ be a graph pattern clause system, and let $p$ be a unary predicate symbol with $\irank(p)=0$.
The language generated by $(\Gamma,p)$ is defined by
\[
L(\Gamma,p)=\left\{\, G \;\middle|\;
\text{$G$ is a graph and the goal $\leftarrow p(G)$ has a refutation in $\Gamma$}
\,\right\}.
\]
\end{definition}

\begin{definition}[Head frame]\label{def:head-frame}
Let $G=(V,E,H,\phi,\psi,\lambda,\ports,\iota)$ be a graph pattern with variable hyperedges.
The \emph{head frame} of $G$ is the structure obtained from $G$ by forgetting the variable-labeling function $\lambda$ and retaining $(V,E,H,\phi,\psi,\ports,\iota)$.
Thus, the head frame records the ordinary labeled graph structure, the variable-hyperedge occurrences, their ordered port lists, and the interface, but it does not record which variable labels are assigned to the variable hyperedges.
If $G$ is ground, then its head frame has no variable-hyperedge occurrences; in particular, facts are also covered by this notion.

Two head frames are said to have the same frame type if they are isomorphic as graph-pattern structures preserving vertex labels, edge labels, variable-hyperedge occurrences, port orders, and interfaces.
\end{definition}

\begin{definition}[Bounded FICS]\label{def:bounded-fics}
Let $m,s,t,w,d$ be nonnegative integers.
A fixed-interface clause system $\Gamma$ is called an \emph{$(m,s,t,w,d)$-bounded fixed-interface clause system} if the following conditions hold:
\begin{enumerate}[label=\textup{(\arabic*)}]
\item the number of clauses in $\Gamma$ is at most $m$,
\item for every variable label $x$ occurring in $\Gamma$, we have $\rank(x)\le w$,
\item for every clause $A \leftarrow B_1,\ldots,B_\ell$ in $\Gamma$,
      \begin{enumerate}[label=\rm(\alph*)]
      \item the number of variable hyperedges in the head pattern of $A$ is at most $s$,
      \item the number $\ell$ of body atoms is at most $t$,
      \item every graph pattern occurring in the clause has maximum vertex degree at most $d$, and
      \item the head frame of the head pattern of $A$, including the ground case where $A$ is the head of a fact, belongs to a fixed finite family $\mathcal H_d$ of head-frame types, where $\mathcal H_d$ depends only on $d$ and the fixed label alphabets.
      \end{enumerate}
\end{enumerate}
\end{definition}

\begin{remark}
The finite head-frame condition in Definition~\ref{def:bounded-fics} is part of the bounded clause-shape assumption.
The bounds on the number of variable hyperedges, body atoms, interface ranks, and vertex degrees do not by themselves bound the number of possible ordinary head frames.
For example, a head pattern may contain arbitrarily long ordinary paths whose vertices are neither interface vertices nor ports of variable hyperedges.

The learning algorithm enumerates bounded fixed-interface clause candidates over the current predicate representatives.
Hence, after the predicate basis has stabilized, the eventual soundness proof relies on there being only finitely many clause candidates.
The finite family $\mathcal H_d$ ensures this finiteness.
This assumption restricts only the finite set of local clause frames from which hypotheses are formed; graphs generated by substituting arbitrary graphs with interfaces for variable hyperedges are not bounded in size by this condition.

For example, one may take $\mathcal H_d$ to be any prescribed finite collection of head-frame types satisfying the degree bound $d$, or the collection of all such types whose number of ordinary vertices and variable-hyperedge occurrences are bounded by fixed constants.
In the present learning theorem, $\mathcal H_d$ is treated as part of the parameterized description of the target class; the result does not depend on a particular canonical choice of this family.
\end{remark}

\begin{definition}[Generated graph class]
Let $m,s,t,w,d,\Delta$ be nonnegative integers.
We denote by
\[
\FICSL_{\Delta}(m,s,t,w,d)
\]
the class of all graph languages $L(\Gamma,p)$ such that
\begin{enumerate}[label=\textup{(\arabic*)}]
\item $\Gamma$ is an $(m,s,t,w,d)$-bounded fixed-interface clause system,
\item $p$ is a unary predicate symbol of $\Gamma$ with $\irank(p)=0$, and
\item every graph in $L(\Gamma,p)$ has maximum degree at most $\Delta$.
\end{enumerate}
\end{definition}

\begin{definition}[Membership problem]
Let $\Gamma$ be a fixed-interface clause system and let $p$ be a unary predicate symbol of $\Gamma$ with $\irank(p)=0$.
The membership problem for $(\Gamma,p)$ is the following decision problem:

\medskip
\textsc{Membership Problem for $(\Gamma,p)$}

\smallskip
\emph{Instance:} A graph $G$.

\emph{Question:} Does $G$ belong to $L(\Gamma,p)$\,?
\end{definition}

\begin{proposition}[Membership bound]\label{prop:membership-bound}
Let $\Gamma$ be an $(m,s,t,w,d)$-bounded fixed-interface clause system, and let $p$ be a unary predicate symbol of $\Gamma$ with $\irank(p)=0$.
Assume that every graph in $L(\Gamma,p)$ has maximum degree at most $\Delta$.
Then the membership problem for $(\Gamma,p)$ is solvable in
\[
O\!\left(\xi(m,s,t,w,\Delta)\, N^{2s(w+1)} \cdot \ISO(N)\right),
\]
where $N$ is the number of vertices of the input graph, $\xi$ is a function depending only on $m,s,t,w,\Delta$, and $\ISO(N)$ denotes the time required to test isomorphism, preserving vertex labels, edge labels, and interfaces, of two graphs with interfaces on at most $N$ vertices.
\end{proposition}

\begin{proof}
The corresponding membership bound in \cite{shoudai2023pac} is stated for the class $\VHFGS_{k,\Delta}(m,s,t,r,w,d)$.
However, the counting argument used there depends on the bounded-degree assumption and on the bounded ranks of interfaces, and does not make essential use of the treewidth bound $k$.
Therefore, in the present fixed-interface setting, we verify the required membership bound directly.

Let $G=(V,E,\phi,\psi)$ be an input graph with $N=|V|$.
We show that the membership problem for $(\Gamma,p)$ can be solved by reducing it to a finite derivation problem over suitably instantiated ground clauses.
Since $\Gamma$ is a fixed-interface clause system, every body atom has a star graph pattern.
Thus, in a derivation of $\leftarrow p(G)$, every variable hyperedge occurring in a head pattern is instantiated by a graph with interface.
By the boundedness assumption on $\Gamma$, every variable label occurring in $\Gamma$ has rank at most $w$.
Hence, when we test membership of the input graph $G$, it is sufficient to consider subgraphs with interface of $G$ whose boundary rank is at most $w$.

Let $\Sub_w(G)$ denote the set of all subgraphs with interfaces of $G$ obtained from valid ordered boundary specifications of rank at most $w$.
We estimate the size of $\Sub_w(G)$.
Fix $r\leq w$. The number of ordered boundary tuples $\beta=(b_1,\ldots,b_r)$ of distinct vertices of $G$ is at most $N^r$. For a fixed such tuple, put $B=\{b_1,\ldots,b_r\}$.
Since every graph in $L(\Gamma,p)$ has maximum degree at most $\Delta$, we only need to consider input graphs of maximum degree at most $\Delta$; if the input graph has maximum degree greater than $\Delta$, then it is immediately rejected.
Hence the number of edges of $G$ incident with vertices in $B$ is at most $r\Delta$.
Therefore the number of possible boundary edge sets $E_B$ is at most $2^{r\Delta}$.
For each candidate pair $(\beta,E_B)$, validity can be checked and, when valid, the corresponding subgraph with interface $K_G(\beta,E_B)$ can be constructed.
Consequently, for fixed $r$, the number of candidates is at most $N^r 2^{r\Delta}$.
Summing over $r=0,1,\ldots,w$, we obtain
\[
|\Sub_w(G)|
\leq
\sum_{r=0}^{w} N^r 2^{r\Delta}
\leq
\xi_1(w,\Delta)\,N^{w+1}
\]
for some function $\xi_1$ depending only on $w$ and $\Delta$.
The exponent $w+1$ is a harmless coarse bound; in fact the preceding estimate gives a slightly sharper polynomial bound.

Now consider a clause of $\Gamma$.
By the definition of an $(m,s,t,w,d)$-bounded fixed-interface clause system, the head pattern of every clause contains at most $s$ variable hyperedges.
Thus, once the underlying clause is fixed, the number of possible substitutions that replace all head variables by members of
$\Sub_w(G)$ is at most $|\Sub_w(G)|^s$.
Since $\Gamma$ has at most $m$ clauses, the total number of instantiated ground clauses that need to be considered is at most $m\,|\Sub_w(G)|^s$.

From these instantiated clauses, we construct a finite ground clause set $\Gamma_G$.
Each atom occurring in $\Gamma_G$ has a ground graph with interface built from a subgraph with interface of $G$, and two such atoms are identified up to isomorphism preserving vertex labels, edge labels, and interfaces.
Hence every required equality test between instantiated atoms can be carried out in time $\ISO(N)$.
The number of distinct ground atoms that can occur in $\Gamma_G$ is bounded by a constant multiple of $|\Sub_w(G)|^s$, since each instantiated head or body atom is determined, up to interface-preserving isomorphism, by a clause of $\Gamma$ and by the choices of substitutions for at most $s$ variable hyperedges.

We then compute the least set of derivable ground atoms from $\Gamma_G$ by iteratively applying modus ponens.
At each application step, one chooses a ground clause and checks whether all atoms in its body have already been derived.
Since every clause body has at most $t$ atoms, and since the total number of ground clauses is at most $m|\Sub_w(G)|^s$, the total number of such checks is bounded by $\xi_2(m,s,t)\,|\Sub_w(G)|^{2s}$ for some function $\xi_2$ depending only on $m,s,t$.
Each check requires only isomorphism tests on graphs with interfaces of size at most $N$, and therefore costs at most $\ISO(N)$ time.

Consequently, the whole derivation procedure is completed in time
\[
O\!\left(
\xi_2(m,s,t)\,|\Sub_w(G)|^{2s}\cdot \ISO(N)
\right).
\]
Substituting the bound
$
|\Sub_w(G)|\leq \xi_1(w,\Delta)\,N^{w+1}
$
gives
\[
O\!\left(
\xi(m,s,t,w,\Delta)\, N^{2s(w+1)}\cdot \ISO(N)
\right),
\]
where $\xi$ is a function depending only on $m,s,t,w,\Delta$.

Finally, $\leftarrow p(G)$ is refutable in $\Gamma$ if and only if the corresponding ground atom is derivable from $\Gamma_G$.
Hence the membership problem for $(\Gamma,p)$ is solvable within the claimed time bound.
\end{proof}

\begin{corollary}\label{cor:membership-bound}
Let $m,s,t,w,d,\Delta$ be fixed constants.
Then the membership problem for $(\Gamma,p)$ is solvable in polynomial time in the size of the input graph.
\end{corollary}

\begin{proof}
By Proposition~\ref{prop:membership-bound}, the membership problem is solvable in
\[
O\!\left(\xi(m,s,t,w,\Delta)\, N^{2s(w+1)}\cdot \ISO(N)\right).
\]
Since $m,s,t,w,\Delta$ are fixed, the factor $\xi(m,s,t,w,\Delta)$ is a constant and the exponent $2s(w+1)$ is fixed.
Moreover, by Luks' algorithm~\cite{luks1982boundedvalence}, isomorphism of bounded-valence graphs is testable in polynomial time.
The ordered interfaces have bounded rank, so preserving them only imposes a fixed-size constraint.
Hence the membership problem is solvable in polynomial time in the size of the input graph.
\end{proof}
\section{Learning Model}\label{sec:learning-model}

We now fix the learning model. The learner receives a positive presentation of an unknown target language and has access to membership queries for that language.
Its task is to output a sequence of fixed-interface clause systems that eventually generate exactly the target language.
We also use the standard polynomial-time update requirement, measured with respect to the cumulative size of the observed positive examples.

\begin{definition}[Positive presentation]
Let $L$ be a graph language, that is, a set of graphs. An infinite sequence $G_1,G_2,\ldots$ of graphs is called a \emph{positive presentation} of $L$ if every $G_i$ belongs to $L$ and every graph in $L$ appears at least once in the sequence.
\end{definition}

\begin{definition}[Membership query]
Let $L$ be a graph language.
A \emph{membership query} for $L$ is a query of the form $\text{``Is } G \in L \text{?''}$ where $G$ is a graph.
The answer is either \textsc{yes} or \textsc{no}.
\end{definition}

\begin{definition}[Learner]
A \emph{learner} is an algorithm that, given a finite initial segment $G_1,\ldots,G_n$ of a positive presentation of a target language $L_*$ and access to membership queries for $L_*$, outputs a hypothesis $(\Gamma_n,p_n)$, where $\Gamma_n$ is a fixed-interface clause system and $p_n$ is a predicate symbol of $\Gamma_n$ with $\irank(p_n)=0$.
\end{definition}

\begin{definition}[Identification in the limit]
Let $\mathcal{C}$ be a class of graph languages.
We say that $\mathcal{C}$ is \emph{identifiable in the limit from positive data and membership queries} if there exists a learner $\mathcal{A}$ such that, for every target language $L_* \in \mathcal{C}$ and every positive presentation $G_1,G_2,\ldots$ of $L_*$, the sequence of hypotheses $(\Gamma_1,p_1),(\Gamma_2,p_2),\ldots$ output by $\mathcal{A}$ satisfies the following condition: there exists a positive integer $N$ such that $L(\Gamma_n,p_n)=L_*$ for all $n \ge N$.
\end{definition}

\begin{definition}[Polynomial update]\label{def:poly-update}
A learner $\mathcal{A}$ is said to have polynomial-time update if there exists a polynomial $P_{\mathrm{upd}}$ such that, for every target language $L_*$, every positive presentation $G_1,G_2,\ldots,$ and every $n \ge 1$, the time required by $\mathcal{A}$ to output $(\Gamma_n,p_n)$ after reading $G_1,\ldots,G_n$ is at most $P_{\mathrm{upd}}\!\left(\sum_{i=1}^n |V(G_i)|\right)$.
\end{definition}

\begin{remark}
Membership queries are treated as oracle calls. Thus, in the polynomial-time update requirement, the cost of answering a membership query is not charged to the learner.
We count the number of membership queries and require that the graphs submitted as queries have size polynomially bounded in the cumulative size of the observed positive examples.
\end{remark}
\section{Main Results}\label{sec:main-results}
\subsection{Learning algorithm}\label{subsec:learning-algorithm}

Before presenting the algorithm, we explain the role of the two finite sets $\mathcal F$ and $\mathcal R$. Both consist of ordered boundary representations extracted from the positive examples observed so far.
The set $\mathcal F$ serves as the predicate basis of the current hypothesis: for each $\rho\in\mathcal F\cup\{\rho_\emptyset\}$, the construction introduces a predicate symbol $p_\rho$.
The set $\mathcal R$ is used as a residual test set for clause admission.
For an ordered boundary representation $\rho$, we write $\mathrm{can}(\rho)$ for the corresponding graph with interface.
The observation table records, for $\rho\in\mathcal F$ and $\sigma\in\mathcal R$, whether $\mathrm{can}(\rho)\odot\mathrm{can}(\sigma)$ belongs to the target language.
Using this table and membership queries, the learner constructs $\Gamma(\mathcal F,\mathcal R)$ by admitting exactly those bounded fixed-interface clause candidates that pass the clause-admission test.

Algorithm~\ref{alg:learning} updates $\mathcal F$ only when the current hypothesis fails to cover the observed positive examples, while $\mathcal R$ is recomputed from the current sample at every stage.
Thus even when the predicate basis is unchanged, the hypothesis may change because the clause-admission tests are re-executed against the enlarged residual set.

\begin{definition}[Composition along interfaces]
Let $G=(V_G,E_G,\phi_G,\psi_G,\iota_G)\:\text{and}\:H=(V_H,E_H,\phi_H,\psi_H,\iota_H)$ be graphs with interfaces.
If $|\iota_G|\ne|\iota_H|$, then $G\odot H$ is undefined.
Otherwise, let $|\iota_G|=|\iota_H|=r$, and write $\iota_G=(u_1,\ldots,u_r)\:\text{and}\:\iota_H=(v_1,\ldots,v_r)$.

Let $\widetilde{G}$ and $\widetilde{H}$ be disjoint copies of $G$ and $H$, respectively, and glue them by identifying the copy of $u_i$ with the copy of $v_i$ for each $i=1,\ldots,r$.
The composition is defined only if the following label-compatibility conditions hold: $\phi_G(u_i)=\phi_H(v_i)\:(i=1,\ldots,r)$, and whenever two ordinary edges, one from $G$ and one from $H$, become the same edge after the identifications, their edge labels coincide.
If these conditions hold, the resulting labeled graph is denoted by $G\odot H$.
Its vertex labels and edge labels are induced from those of $G$ and $H$; the compatibility conditions ensure that these induced
labeling functions are well-defined. If one of the compatibility conditions fails, then $G\odot H$ is undefined.
\end{definition}

\begin{remark}
In the composition $G\odot H$, the interface vertices are used only for gluing.
Hence the result is regarded as a graph with interface only when an interface is specified separately.
In particular, when we use $G\odot H$ in the definitions of contexts and substructures below, it will usually be treated simply as a graph.
\end{remark}

\begin{definition}[Context of a predicate]\label{def:context}
Let $\Gamma$ be a fixed-interface clause system, let $p$ be the designated start predicate symbol of $\Gamma$ with $\irank(p)=0$, and let $q$ be a predicate symbol of $\Gamma$.
Let $G$ be a graph with interface such that $|\iota_G|=\irank(q)$.
We say that $G$ is a \emph{context} of $q$ with respect to $(\Gamma,p)$ if, for every graph with interface $K$ such that $|\iota_K|=\irank(q)$,
\[
\leftarrow q(K) \text{ has a refutation in } \Gamma
\quad\Longleftrightarrow\quad
\leftarrow p(G\odot K) \text{ has a refutation in } \Gamma,
\]
whenever the composition $G\odot K$ is defined.
\end{definition}

\begin{definition}[Finite context property (FCP)]\label{def:finite-context-property}
Let $\Gamma$ be a fixed-interface clause system, and let $p$ be a unary predicate symbol of $\Gamma$ with $\irank(p)=0$.
We say that $(\Gamma,p)$ has the \emph{finite context property} if, for every unary predicate symbol $q$ of $\Gamma$, there exists a graph with interface $G_q$ such that $G_q$ is a context of $q$ with respect to $(\Gamma,p)$.
\end{definition}

\begin{definition}[Context-learnable graph languages]\label{def:context-learnable-graph-languages}
Let $m,s,t,w,d,\Delta$ be nonnegative integers.
We denote by
\[
\FICSLFCP_{\Delta}(m,s,t,w,d)
\]
the class of all graph languages $L(\Gamma,p)$ such that
\begin{enumerate}[label=\textup{(\arabic*)}]
\item $\Gamma$ is an $(m,s,t,w,d)$-bounded fixed-interface clause system,
\item $p$ is a unary predicate symbol of $\Gamma$ with $\irank(p)=0$,
\item $(\Gamma,p)$ has the finite context property,
\item every graph in $L(\Gamma,p)$ has maximum degree at most $\Delta$.
\end{enumerate}
\end{definition}

\begin{remark}[Interface-degree-safe clause systems]\label{rem:interface-degree-safe}
The bounded-degree condition, namely condition~(4) of Definition~~\ref{def:context-learnable-graph-languages}, is a semantic condition on the generated language. We describe here a useful syntactic sufficient condition for it.

For a fixed-interface clause system $\Gamma$, define the predicate dependency graph of $\Gamma$ as the directed graph whose vertices are the predicate symbols of $\Gamma$, and in which there is an edge from $p$ to $q$ if $\Gamma$ contains a clause with head predicate $p$ and a body atom with predicate $q$.
Let $A \leftarrow B_1,\ldots,B_m$ be a fixed-interface clause of $\Gamma$, and let $p$ be the predicate symbol of the head atom $A$.
For a body atom $B_i$ with predicate symbol $q_i$, we call $B_i$ recursive in this clause if $p$ and $q_i$ belong to the same strongly connected component of the predicate dependency graph of $\Gamma$.

Let $H$ be the head pattern of $A$.
We call the clause \emph{interface-degree-safe} if every vertex of $H$ that belongs both to the interface of $H$ and to the port list of a variable hyperedge corresponding to a recursive body atom satisfies the following two conditions:
it occurs in the port list of at most one such recursive variable hyperedge, and it is incident with no ordinary edge of $H$.
Thus, an interface vertex that is passed through a recursive call cannot also accumulate new ordinary edges at the same derivation step, nor can it be shared by several recursive components. Non-recursive body atoms may still be attached at interface vertices; such attachments do not by themselves propagate degree growth along a recursive cycle. A fixed-interface clause system is called \emph{interface-degree-safe} if all its clauses are interface-degree-safe.

For interface-degree-safe systems, recursive derivations cannot accumulate unbounded degree at vertices passed through interfaces, and hence the generated graphs have maximum degree bounded by a constant depending only on the clause system.
In the sequel, however, the learning argument uses only the existence of a fixed degree bound $\Delta$, not this particular sufficient condition.
\end{remark}

\begin{example}
Figure~\ref{fig:fcp-example} illustrates the finite context property for a simple fixed-interface clause system.
The start predicate is $p$, with $\irank(p)=0$, while $q$ and $r$ are auxiliary unary predicate symbols with interface rank $2$.
The predicate $r$ generates path-like interface graphs, $q$ combines such interface graphs, and $p$ derives graphs with the empty interface from $q$-structures.
For each auxiliary predicate $s\in\{q,r\}$, the derivability of $s(K)$ can be tested, for every graph with interface $K$, by placing $K$ into a suitable context $C_s$ and asking whether the resulting graph is derivable from the start predicate $p$.
Thus the auxiliary predicates are represented by contexts with respect to $p$, as required by the finite context property.

\begin{figure}[t]
  \centering
  \includegraphics[scale=0.48]{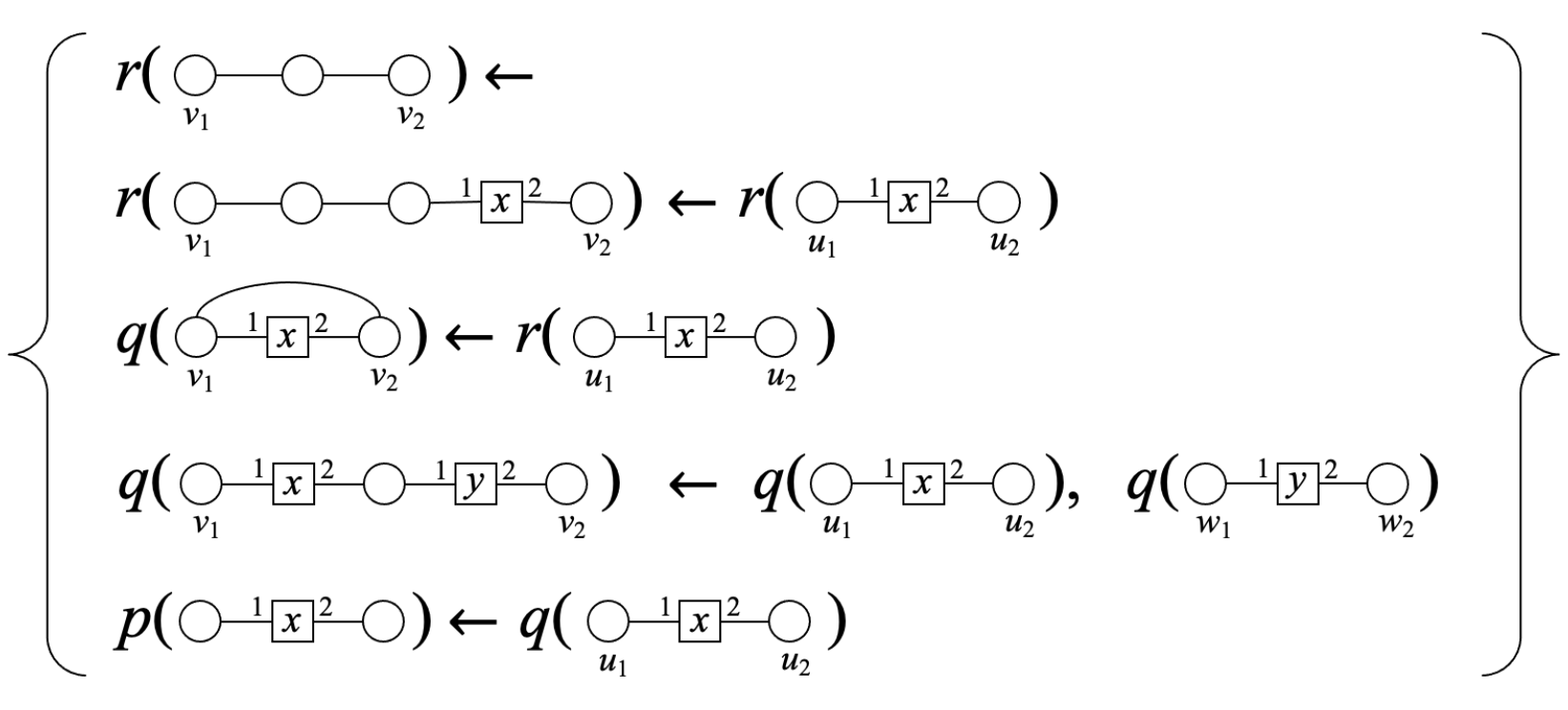}
  \caption{
  A fixed-interface clause system illustrating the finite context property.
  The start predicate $p$ has interface rank $0$, whereas the auxiliary
  predicates $q$ and $r$ have interface rank $2$.
  Vertex names are local to each displayed clause and indicate interface positions only.
  }
  \label{fig:fcp-example}
\end{figure}

Both fixed-interface clause systems illustrated in Figures~\ref{fig:fics-repeated-variable} and~\ref{fig:fcp-example} satisfy the interface-degree-safety condition of Remark~\ref{rem:interface-degree-safe}.
\end{example}

\begin{definition}[Context candidates]\label{def:context-candidates}
Let $D$ be a nonempty finite set of graphs.
We define $\Con(D)$ to be the set of all graphs with interface $G$ for which there exists a graph with interface $H$ such that $G\odot H$ is defined and $G\odot H\in D$.
\end{definition}

\begin{definition}[Characteristic sample]\label{def:characteristic-sample}
Let $(\Gamma,p)$ be a fixed-interface clause system with the finite context property.
A finite subset $D \subseteq L(\Gamma,p)$ is called \emph{characteristic for contexts} if, for every unary predicate symbol $q$ of $\Gamma$, the set $\Con(D)$ contains a context of $q$ with respect to $(\Gamma,p)$.
\end{definition}

\begin{theorem}[Identification in the limit]\label{thm:correctness}
For every target language $L_* \in \FICSLFCP_{\Delta}(m,s,t,w,d)$, Algorithm~\ref{alg:learning} identifies $L_*$ in the limit from positive data and membership queries.
\end{theorem}

\begin{remark}
The proof of the theorem will be given stepwise.
The main ingredients are:
\begin{enumerate}[label=\textup{(\arabic*)}]
\item every target context eventually appears in $\Con(D_n)$,
\item the hypothesis eventually contains predicate representatives corresponding to all target predicates,
\item every target clause is eventually reconstructed in $\Gamma(\mathcal F,\mathcal R)$, and
\item spurious clauses are excluded by membership queries.
\end{enumerate}
\end{remark}

Let $G_{\emptyset}:=(\emptyset,\emptyset,\phi_{\emptyset},\psi_{\emptyset},())$ denote the empty graph with the empty interface, and let $p_{\emptyset}$ be the unary predicate symbol corresponding to $G_{\emptyset}$.

\begin{definition}[Realization of a graph pattern]\label{def:realization-head-pattern}
Let $G=(V,E,H,\phi,\psi,\lambda,\ports,\iota)$ be a graph pattern with variable hyperedges, and let $\theta=\{x_1:=K_1,\ldots,x_m:=K_m\}$ be a substitution, where each $K_i$ is a graph with interface and $\rank(x_i)=|\iota_{K_i}|$.
Assume that every variable label occurring in $G$ belongs to $\{x_1,\ldots,x_m\}$.

The \emph{realization} of $G$ by $\theta$, denoted by $\Real(G,\theta)$, is the graph with interface obtained from $G$ by simultaneously replacing every hyperedge labeled $x_i$ by a fresh copy of $K_i$, identifying the $j$-th vertex of the port list of that hyperedge with the $j$-th interface vertex of the copy of $K_i$ for each $j=1,\ldots,\rank(x_i)$, and finally keeping the interface $\iota$ of $G$ as the interface of the resulting graph.

Equivalently, if $G\theta$ denotes the graph pattern obtained by applying $\theta$ to $G$, then $\Real(G,\theta)$ is the graph with interface underlying $G\theta$.
When $G$ is the head pattern of a clause, we also refer to $\Real(G,\theta)$ as the realization of the head pattern by $\theta$.
\end{definition}

Throughout Algorithms~\ref{alg:construct-obrep}--\ref{alg:learning}, we write $\mathrm{Oracle}_{L_*}(G)$ for the result of the membership query ``$G\in L_*$?'', which returns \textbf{true} or \textbf{false}.
For an ordered boundary representation $\rho$, we write $\mathrm{rank}(\rho)$ for the length of its ordered boundary.
Whenever a representative graph with interface is needed, we write $\mathrm{can}(\rho)$ for the subgraph with interface determined by $\rho$.
The finite family $\mathcal H_d$ is fixed throughout the construction and is used as a global parameter when constructing $\Gamma(F,R)$ in Algorithm~\ref{alg:construct-gamma}.

\begin{algorithm}[t]
\caption{Construction of ordered boundary representations}
\label{alg:construct-obrep}
\begin{algorithmic}[1]
\Require A finite positive sample $D$ and an interface-rank bound $w$
\Ensure A finite set $\mathrm{BRep}_w(D)$ of ordered boundary representations

\State $\mathrm{BRep}_w(D)\leftarrow\emptyset$

\ForAll{$X=(V_X,E_X,\phi_X,\psi_X,())\in D$}
    \For{$r=0,1,\ldots,w$}
        \State $\mathcal B_r(X)\leftarrow
        \{(b_1,\ldots,b_r) \mid b_1,\ldots,b_r
        \text{ are distinct vertices of }X\}$
        \ForAll{$\beta=(b_1,\ldots,b_r)\in\mathcal B_r(X)$}
            \State $B\leftarrow\{b_1,\ldots,b_r\}$
            \State $\mathrm{Inc}_X(B)\leftarrow
            \{e\in E_X\mid e\cap B\neq\emptyset\}$
            \ForAll{$E_B\subseteq \mathrm{Inc}_X(B)$}
                \State \Comment{$(\beta,E_B)$ is valid since $E_B\subseteq \mathrm{Inc}_X(B)$}
                \State $K\leftarrow K_X(\beta,E_B)$
                \State $\rho\leftarrow(\beta,E_B)$
                \State $\mathrm{can}(\rho)\leftarrow K$
                \State Add $\rho$ to $\mathrm{BRep}_w(D)$
            \EndFor
        \EndFor
    \EndFor
\EndFor

\State \Return $\mathrm{BRep}_w(D)$
\end{algorithmic}
\end{algorithm}

\begin{algorithm}[t]
\caption{Construction of the observation table
$\mathrm{Obs}_{\mathcal F,\mathcal R}$}
\label{alg:construct-observation-table}
\begin{algorithmic}[1]
\Require Finite sets $\mathcal F$ and $\mathcal R$ of ordered boundary representations
\Ensure A Boolean table $\mathrm{Obs}_{\mathcal F,\mathcal R}$ on
$\mathcal F\times\mathcal R$

\State Initialize
$\mathrm{Obs}_{\mathcal F,\mathcal R}(\rho,\sigma)\leftarrow\mathrm{false}$
for all $(\rho,\sigma)\in\mathcal F\times\mathcal R$

\ForAll{$\rho\in\mathcal F$}
    \ForAll{$\sigma\in\mathcal R$}
        \State Let $G\leftarrow\mathrm{can}(\rho)$ and
        $K\leftarrow\mathrm{can}(\sigma)$
        \If{$G\odot K$ is defined \textbf{and}
        $\mathrm{Oracle}_{L_*}(G\odot K)$}
            \State
            $\mathrm{Obs}_{\mathcal F,\mathcal R}(\rho,\sigma)
            \leftarrow\mathrm{true}$
        \EndIf
    \EndFor
\EndFor

\State \Return $\mathrm{Obs}_{\mathcal F,\mathcal R}$
\end{algorithmic}
\end{algorithm}

\begin{algorithm}[p]
\caption{Construction of $\Gamma(\mathcal F,\mathcal R)$}
\label{alg:construct-gamma}
\begin{algorithmic}[1]
\Require Finite sets $\mathcal F$ and $\mathcal R$ of ordered boundary representations
\Ensure A fixed-interface clause system $\Gamma(\mathcal F,\mathcal R)$

\State $\Gamma(\mathcal F,\mathcal R)\leftarrow\emptyset$

\ForAll{$\rho\in\mathcal F\cup\{\rho_\emptyset\}$}
    \State Introduce $p_\rho$ with
    $\mathrm{irank}(p_\rho)=\mathrm{rank}(\rho)$
\EndFor

\State Construct $\mathrm{Obs}_{\mathcal F,\mathcal R}$ by Algorithm~\ref{alg:construct-observation-table}

\State $\mathcal C\leftarrow$ the set of all bounded fixed-interface clause candidates over $\mathcal F\cup\{\rho_\emptyset\}$ whose head frames belong to $\mathcal H_d$

\ForAll{$\gamma=p_\rho(H)\leftarrow
p_{\rho_1}(H_1),\ldots,p_{\rho_\ell}(H_\ell)\in\mathcal C$}
    \State \Comment{Clause-admission test for $\gamma$}
    \If{$\ell=0$}
        \State Let $G\leftarrow\mathrm{can}(\rho)$
        \If{$G\odot H$ is defined \textbf{and}
        $\mathrm{Oracle}_{L_*}(G\odot H)$}
            \State Add $\gamma$ to $\Gamma(\mathcal F,\mathcal R)$
        \EndIf
    \Else
        \State Let $x_i$ be the variable label of $H_i$ for each $i=1,\ldots,\ell$
        \State Let $X(\gamma)=\{x_i\mid 1\leq i\leq \ell\}$
        \State $\mathrm{admissible}\leftarrow\mathrm{true}$

        \ForAll{$(\sigma_x)_{x\in X(\gamma)}\in\mathcal R^{X(\gamma)}$}
            \If{$\mathrm{Obs}_{\mathcal F,\mathcal R}(\rho_i,\sigma_{x_i})
            =\mathrm{true}$ for all $i=1,\ldots,\ell$}
                \State Let $K_x\leftarrow\mathrm{can}(\sigma_x)$
                for each $x\in X(\gamma)$
                \If{$K\leftarrow
                \Real\bigl(H,\{x:=K_x\}_{x\in X(\gamma)}\bigr)$
                is defined}
                    \If{$\neg\mathrm{Oracle}_{L_*}
                    \bigl(\mathrm{can}(\rho)\odot K\bigr)$}
                        \State $\mathrm{admissible}\leftarrow\mathrm{false}$
                        \State \textbf{break}
                    \EndIf
                \EndIf
            \EndIf
        \EndFor

        \If{$\mathrm{admissible}$}
            \State Add $\gamma$ to $\Gamma(\mathcal F,\mathcal R)$
        \EndIf
    \EndIf
\EndFor

\State \Return $\Gamma(\mathcal F,\mathcal R)$
\end{algorithmic}
\end{algorithm}

\begin{remark}\label{rem:recomputation}
At each stage $n$, the clause set $\Gamma(\mathcal F,\mathcal R)$ is recomputed from the current sets $\mathcal F$ and $\mathcal R$.
In particular, even when $\mathcal F$ is not updated, the residual set $\mathcal R$ is updated to $\mathrm{BRep}_w(D_n)$, and the observation table $\mathrm{Obs}_{\mathcal F,\mathcal R}$ and the clause-admission tests in Algorithm~\ref{alg:construct-gamma} are recomputed with respect to the new $\mathcal R$.
Thus the hypothesis at stage $n$ is determined by the current pair $(\mathcal F,\mathcal R)$, not merely by the current predicate basis $\mathcal F$.
\end{remark}

\begin{algorithm}[t]
\caption{Learning algorithm}
\label{alg:learning}
\begin{algorithmic}[1]
\Require A positive presentation $G_1,G_2,\ldots$ of the target language
$L_*$, an interface-rank bound $w$, and access to membership queries for
$L_*$
\Ensure A sequence of hypotheses $(\Gamma_n,p_{\rho_\emptyset})$

\State $\mathcal F\leftarrow\emptyset$
\State $\mathcal R\leftarrow\emptyset$

\For{$n=1,2,3,\ldots$}
    \State $D_n\leftarrow\{G_1,\ldots,G_n\}$

    \State Construct $\Gamma'=\Gamma(\mathcal F,\mathcal R)$ by Algorithm~\ref{alg:construct-gamma}

    \If{$D_n\nsubseteq L(\Gamma',p_{\rho_\emptyset})$}
        \State Construct $\mathcal F=\mathrm{BRep}_w(D_n)$ by Algorithm~\ref{alg:construct-obrep}
    \EndIf

    \State Construct $\mathcal R=\mathrm{BRep}_w(D_n)$ by Algorithm~\ref{alg:construct-obrep}

    \State Construct $\Gamma_n=\Gamma(\mathcal F,\mathcal R)$ by Algorithm~\ref{alg:construct-gamma}

    \State Output $(\Gamma_n,p_{\rho_\emptyset})$
\EndFor
\end{algorithmic}
\end{algorithm}
\subsection{Identification in the limit}\label{subsec:correctness}

We now prove Theorem~\ref{thm:correctness}.
The proof is organized around the eventual appearance of target contexts and the eventual elimination of spurious clauses.

First, we show that every relevant target context eventually appears among the ordered boundary representations extracted from the observed positive examples.
This yields predicate representatives for all target predicates at sufficiently late update stages.
Next, we prove that every target clause is then reconstructed over these representatives, and that refutations in the target system can be simulated in the reconstructed hypothesis.
The remaining part of the proof establishes eventual soundness.
Once the predicate basis is fixed, every spurious non-fact clause is eventually exposed by some residual representative and is rejected by a membership query.
Hence, after sufficiently many examples, the hypotheses generate no graph outside the target language. Combining this eventual soundness with the simulation of target refutations gives identification in the limit.

\begin{lemma}\label{lem:eventual-context}
Assume that $L_* = L(\Gamma_*,p_*) \in \FICSLFCP_{\Delta}(m,s,t,w,d)$, and let $G_1,G_2,\ldots$ be a positive presentation of $L_*$.
For each $n \ge 1$, let $D_n=\{G_1,\ldots,G_n\}$.
Let $q$ be a unary predicate symbol of $\Gamma_*$ such that there exists a graph with interface $K_q$ for which $\leftarrow q(K_q)$ has a refutation in $\Gamma_*$.
Then there exist a context $C_q$ of $q$ with respect to $(\Gamma_*,p_*)$ and a positive integer $n_q$ such that $C_q \in \Con(D_n)$ for every $n \ge n_q$.
\end{lemma}

\begin{proof}
Since $(\Gamma_*,p_*)$ has the finite context property, there exists a graph with interface $C_q$ that is a context of $q$ with respect to $(\Gamma_*,p_*)$.
By assumption, $\leftarrow q(K_q)$ has a refutation in $\Gamma_*$.
Since $C_q$ is a context of $q$ with respect to $(\Gamma_*,p_*)$, it follows from Definition~\ref{def:context} that $\leftarrow p_*(C_q \odot K_q)$ also has a refutation in $\Gamma_*$.
Hence $C_q \odot K_q \in L_*$.

Since $G_1,G_2,\ldots$ is a positive presentation of $L_*$, there exists a positive integer $n_q$ such that $C_q \odot K_q \in D_n$ for every $n \ge n_q$.
By Definition~\ref{def:context-candidates}, this implies $C_q \in \Con(D_n)$ for every $n \ge n_q$.
\end{proof}

\begin{lemma}\label{lem:eventual-boundary-representative}
Assume that $L_* = L(\Gamma_*,p_*) \in \FICSLFCP_{\Delta}(m,s,t,w,d)$, and let $G_1,G_2,\ldots$ be a positive presentation of $L_*$.
For each $n\ge 1$, let $D_n=\{G_1,\ldots,G_n\}$.
Let $q$ be a unary predicate symbol of $\Gamma_*$ such that there exists a graph with interface $K_q$ for which $\leftarrow q(K_q)$ has a refutation in $\Gamma_*$.
Then there exist a context $C_q$ of $q$ with respect to $(\Gamma_*,p_*)$, an ordered boundary representation $\rho_q$, and a positive integer $n_q$ such that, for every $n\ge n_q$, $\rho_q\in \mathrm{BRep}_w(D_n)$ and $\mathrm{can}(\rho_q)=C_q$.
\end{lemma}

\begin{proof}
By Lemma~\ref{lem:eventual-context}, there exist a context $C_q$ of $q$ with respect to $(\Gamma_*,p_*)$ and a positive integer $n_q$ such that $C_q\in \Con(D_n)$ for every $n\ge n_q$.
Fix $n\ge n_q$.
Since $C_q\in \Con(D_n)$, there exists a graph with interface $H_q$ such that $C_q\odot H_q\in D_n$.
Let $X=C_q\odot H_q$.
By the canonical copy of $C_q$ in $X$, we mean the image in $X$ of the fresh copy of $C_q$ used in the construction of the composition.
Write the interface of $C_q$ as $\iota_{C_q}=(b_1,\ldots,b_r)$, and use the same symbols for the corresponding interface vertices of its canonical copy in $X$.

Let $B=\{b_1,\ldots,b_r\}$.
Choose $E_B$ to be the set of all edges of the canonical copy of $C_q$ in $X$ that are incident with at least one vertex of $B$.
Then every edge in $E_B$ is an edge of $X$ incident with a vertex of $B$, and hence $(\beta,E_B)$ is a valid ordered boundary specification in $X$, where $\beta=(b_1,\ldots,b_r)$.
Since the composition $C_q\odot H_q$ is formed from disjoint copies by identifying only the corresponding interface vertices, deleting the vertices in $B$ separates the non-interface part of the copy of $C_q$ from the non-interface part of the copy of $H_q$.
Therefore the vertices of the copy of $C_q$ outside the interface are exactly the vertices reachable in $X-B$ from the non-interface endpoints of the edges in $E_B$.
Consequently, the subgraph with interface $K_X(\beta,E_B)$ constructed by Definition~\ref{def:valid-ordered-boundary-specification} is precisely the canonical copy of $C_q$ in $X$.

Let $\rho_q=(\beta,E_B)$ be this ordered boundary representation.
Since $C_q$ has interface rank at most $w$, Algorithm~\ref{alg:construct-obrep} enumerates $\rho_q$ when it is applied to $D_n$.
Hence $\rho_q\in \mathrm{BRep}_w(D_n)$.
Moreover, by the preceding paragraph and by the construction of Algorithm~\ref{alg:construct-obrep}, the canonical representative associated with $\rho_q$ is precisely the canonical copy of $C_q$.
Therefore $\mathrm{can}(\rho_q)=C_q$.
This proves the lemma.
\end{proof}

\begin{definition}[Update stage]\label{def:update-stage}
Let $(\Gamma_n,p_{\rho_\emptyset})$ be the $n$-th hypothesis output by Algorithm~\ref{alg:learning}.
We say that $n$ is an \emph{update stage} if Algorithm~\ref{alg:learning} executes the update of the predicate basis at stage $n$, that is, if it constructs $\mathcal F=\mathrm{BRep}_w(D_n)$ by Algorithm~\ref{alg:construct-obrep}.
\end{definition}

\begin{lemma}\label{lem:eventual-predicate-representatives}
Assume that $L_* = L(\Gamma_*,p_*) \in \FICSLFCP_{\Delta}(m,s,t,w,d)$, and let $G_1,G_2,\ldots$ be a positive presentation of $L_*$.
For each $n\ge 1$, let $D_n=\{G_1,\ldots,G_n\}$.
If $D_n$ is characteristic for contexts with respect to $(\Gamma_*,p_*)$ and $n$ is an update stage, then for every unary predicate symbol $q$ of $\Gamma_*$ such that there exists a graph with interface $K_q$ for which $\leftarrow q(K_q)$ has a refutation in $\Gamma_*$, the hypothesis $(\Gamma_n,p_{\rho_\emptyset})$ contains a predicate symbol $p_{\rho_q}$ such that $\mathrm{can}(\rho_q)=C_q$, where $C_q$ is a context of $q$ with respect to $(\Gamma_*,p_*)$.
\end{lemma}

\begin{proof}
Since $D_n$ is characteristic for contexts, Definition~\ref{def:characteristic-sample} implies that, for every such unary predicate symbol $q$, the set $\Con(D_n)$ contains a context $C_q$ of $q$ with respect to $(\Gamma_*,p_*)$.
By the preceding lemma, there exists an ordered boundary representation $\rho_q\in \mathrm{BRep}_w(D_n)$ such that $\mathrm{can}(\rho_q)=C_q$.
Since $n$ is an update stage, Algorithm~\ref{alg:learning} sets $\mathcal F=\mathrm{BRep}_w(D_n)$ at stage $n$.
Hence $\rho_q\in\mathcal F$.
By Algorithm~\ref{alg:construct-gamma}, for each $\rho\in \mathcal F\cup\{\rho_\emptyset\}$, the construction introduces a predicate symbol $p_\rho$ with $\mathrm{irank}(p_\rho)=\mathrm{rank}(\rho)$.
Therefore $p_{\rho_q}$ occurs in $\Gamma_n$.
\end{proof}

\begin{definition}[Spurious clause candidate]\label{def:spurious-clause}
Let
\[
\gamma \;=\; p_\rho(H) \leftarrow
p_{\rho_1}(H_1),\ldots,p_{\rho_\ell}(H_\ell)
\]
be a fixed-interface clause candidate with $\ell\geq 1$.
For each $i=1,\ldots,\ell$, let $x_i$ be the unique variable label occurring in the star graph pattern $H_i$, and put $X(\gamma)=\{x_i\mid 1\leq i\leq \ell\}$.
We say that $\gamma$ is \emph{spurious with respect to $L_*$} if there exist graphs with interface $K_x \; (x\in X(\gamma))$ such that $\mathrm{can}(\rho_i)\odot K_{x_i}\in L_*\: (i=1,\ldots,\ell)$, the realization $\Real\bigl(H,\{x:=K_x\}_{x\in X(\gamma)}\bigr)$ is defined, and $\mathrm{can}(\rho)\odot\Real\bigl(H,\{x:=K_x\}_{x\in X(\gamma)}\bigr)\notin L_*$.
If no such family $(K_x)_{x\in X(\gamma)}$ exists, then $\gamma$ is called \emph{non-spurious with respect to $L_*$}.
\end{definition}

\begin{lemma}\label{lem:one-step-preservation-by-nonspurious-clause}
Let $\gamma \;=\; p_\rho(H) \leftarrow p_{\rho_1}(H_1),\ldots,p_{\rho_\ell}(H_\ell)$ be a fixed-interface clause candidate with $\ell\geq 1$.
For each $i=1,\ldots,\ell$, let $x_i$ be the unique variable label occurring in the star graph pattern $H_i$, and put $X(\gamma)=\{x_i\mid 1\leq i\leq \ell\}$.
Assume that $\gamma$ is non-spurious with respect to $L_*$.
Let $(K_x)_{x\in X(\gamma)}$ be a family of graphs with interfaces such that $\mathrm{can}(\rho_i)\odot K_{x_i}\in L_*\: (i=1,\ldots,\ell)$.
Assume that $\Real\bigl(H,\{x:=K_x\}_{x\in X(\gamma)}\bigr)$ is defined.
Then $\mathrm{can}(\rho)\odot\Real\bigl(H,\{x:=K_x\}_{x\in X(\gamma)}\bigr)\in L_*$.
\end{lemma}

\begin{proof}
Assume, toward a contradiction, that
\[
\mathrm{can}(\rho)\odot
\Real\bigl(H,\{x:=K_x\}_{x\in X(\gamma)}\bigr)
\notin L_*.
\]
Then the family $(K_x)_{x\in X(\gamma)}$ witnesses that $\gamma$ is spurious with respect to $L_*$, because $\mathrm{can}(\rho_i)\odot K_{x_i}\in L_*\: (i=1,\ldots,\ell)$, while $\mathrm{can}(\rho)\odot\Real\bigl(H,\{x:=K_x\}_{x\in X(\gamma)}\bigr)\notin L_*$.
This contradicts the assumption that $\gamma$ is non-spurious.
Therefore
\[
\mathrm{can}(\rho)\odot
\Real\bigl(H,\{x:=K_x\}_{x\in X(\gamma)}\bigr)
\in L_*.
\]
\end{proof}

\begin{lemma}\label{lem:eventual-target-clause-reconstruction}
Assume that $L_* = L(\Gamma_*,p_*) \in \FICSLFCP_{\Delta}(m,s,t,w,d)$, and let $G_1,G_2,\ldots$ be a positive presentation of $L_*$. For each $n\ge 1$, let $D_n=\{G_1,\ldots,G_n\}$.
Then there exists an integer $n_1$ such that, for every update stage $n\ge n_1$, the hypothesis $(\Gamma_n,p_{\rho_\emptyset})$ contains predicate symbols corresponding to all target predicate contexts, and every target clause of $\Gamma_*$ is  reconstructed in $\Gamma_n$ as a clause over those predicate representatives.
\end{lemma}

\begin{proof}
By Lemma~\ref{lem:eventual-context}, for every unary predicate symbol $q$ of $\Gamma_*$ for which there exists a graph with interface $K_q$ such that $\leftarrow q(K_q)$ has a refutation in $\Gamma_*$, there exist a context $C_q$ of $q$ with respect to $(\Gamma_*,p_*)$ and an integer $n_q$ such that $C_q\in \Con(D_n)$ for all $n\ge n_q$.

Since $\Gamma_*$ has only finitely many predicate symbols, we may choose an integer $n_1$ such that, for all $n\ge n_1$, the sample $D_n$ is characteristic for contexts.
Let $n\ge n_1$ be an update stage.
Then, by Lemma~\ref{lem:eventual-predicate-representatives}, for every relevant target predicate symbol $q$ of $\Gamma_*$, the hypothesis $(\Gamma_n,p_{\rho_\emptyset})$ contains a predicate symbol $p_{\rho_q}$ such that $\mathrm{can}(\rho_q)=C_q$, where $C_q$ is the chosen context of $q$ with respect to $(\Gamma_*,p_*)$.
Now let $q(H)\leftarrow q_1(H_1),\ldots,q_\ell(H_\ell)$ be any clause of $\Gamma_*$, where $0\le \ell\le t$.
For each predicate symbol $q_i$, let $C_{q_i}$ be the chosen context of $q_i$, and let $\rho_{q_i}$ be its corresponding representative in $\mathcal F$, so that $\mathrm{can}(\rho_{q_i})=C_{q_i}$.
Similarly, let $\rho_q\in\mathcal F$ satisfy $\mathrm{can}(\rho_q)=C_q$.
Since $n$ is an update stage, all these representatives belong to the current predicate basis $\mathcal F$.
Moreover, since $\Gamma_*$ is $(m,s,t,w,d)$-bounded, the head frame of $H$ belongs to the fixed finite family $\mathcal H_d$.
Hence Algorithm~\ref{alg:construct-gamma} considers the bounded fixed-interface clause candidate
\[
p_{\rho_q}(H)\leftarrow
p_{\rho_{q_1}}(H_1),\ldots,p_{\rho_{q_\ell}}(H_\ell).
\]

It remains to show that this candidate is admitted into $\Gamma_n$.

First suppose that $\ell=0$.
Then the target clause is $q(H)\leftarrow$.
Hence $\leftarrow q(H)$ has a refutation in $\Gamma_*$. Since $C_q$ is a context of $q$ with respect to $(\Gamma_*,p_*)$, we have $C_q\odot H\in L_*$.
Equivalently, $\mathrm{can}(\rho_q)\odot H\in L_*$.
Therefore the membership query in the $\ell=0$ case of Algorithm~\ref{alg:construct-gamma} returns \textbf{true}, and the candidate clause $p_{\rho_q}(H)\leftarrow$ is added to $\Gamma_n$.

Next suppose that $\ell\ge 1$. For each $i=1,\ldots,\ell$, let $x_i$ be the unique variable label occurring in the star graph pattern $H_i$, and put $X=\{x_i\mid 1\le i\le \ell\}$.
To prove that the candidate is admitted, it is enough to verify the non-fact admission condition of Algorithm~\ref{alg:construct-gamma}.
Thus we consider an arbitrary family $(\sigma_x)_{x\in X}\in\mathcal R^X$ for which all body observations are positive.
If no such family exists, the admission condition is satisfied vacuously.
Assume therefore that $\mathrm{Obs}_{\mathcal F,\mathcal R}(\rho_{q_i},\sigma_{x_i})=\mathrm{true}$ for all $i=1,\ldots,\ell$.
By the definition of the observation table, this means that $\mathrm{can}(\rho_{q_i})\odot\mathrm{can}(\sigma_{x_i})\in L_*$ for all $i=1,\ldots,\ell$.
Since $\mathrm{can}(\rho_{q_i})=C_{q_i}$, we have $C_{q_i}\odot\mathrm{can}(\sigma_{x_i})\in L_*$ for all $i=1,\ldots,\ell$.
Because $C_{q_i}$ is a context of $q_i$ with respect to $(\Gamma_*,p_*)$, it follows that $\leftarrow q_i(\mathrm{can}(\sigma_{x_i}))$ has a refutation in $\Gamma_*$ for every $i=1,\ldots,\ell$.
Now let $\theta=\{x:=\mathrm{can}(\sigma_x)\}_{x\in X}$.
This is a single substitution on variable labels. Hence if the same variable label occurs in several body atoms or several hyperedges of the head pattern, all such occurrences are replaced by the same graph with interface.
Assume that $K=\Real(H,\theta)$ is defined.
Applying the target clause $q(H)\leftarrow q_1(H_1),\ldots,q_\ell(H_\ell)$ in $\Gamma_*$ with the substitution $\theta$, and using the refutations of the body goals obtained above, we obtain a refutation of $\leftarrow q(K)$ in $\Gamma_*$.
Since $C_q$ is a context of $q$ with respect to $(\Gamma_*,p_*)$, this implies $C_q\odot K\in L_*$.
Equivalently, $\mathrm{can}(\rho_q)\odot K\in L_*$.
Thus every membership query required by the non-fact admission test in Algorithm~\ref{alg:construct-gamma} is answered positively.
Therefore the candidate clause $p_{\rho_q}(H)\leftarrow p_{\rho_{q_1}}(H_1),\ldots,p_{\rho_{q_\ell}}(H_\ell)$ is admitted into $\Gamma_n$.

Since the target clause was arbitrary, every target clause of $\Gamma_*$ is reconstructed in $\Gamma_n$ over the predicate representatives given by the target contexts.
\end{proof}

\begin{lemma}\label{lem:refutation-simulation}
Assume the same setting as in Lemma~\ref{lem:eventual-target-clause-reconstruction}.
Let $n$ be an update stage at which every target clause is reconstructed in $\Gamma_n$.
If $\leftarrow q(K)$ has a refutation in $\Gamma_*$, then $\leftarrow p_{\rho_q}(K)$ has a refutation in $\Gamma_n$, where $p_{\rho_q}$ is the predicate representative of $q$ used in $\Gamma_n$.
\end{lemma}

\begin{proof}
We prove the claim by induction on the refutation tree in $\Gamma_*$.

Consider the first target clause applied in the refutation.
By Lemma~\ref{lem:eventual-target-clause-reconstruction}, this clause has a reconstructed counterpart in $\Gamma_n$, obtained by replacing each target predicate symbol $q'$ with its corresponding predicate representative $p_{\rho_{q'}}$.
The graph patterns in the head and body are unchanged.
Thus the same substitution on variable labels used in the target derivation step can be used in the reconstructed clause.
In particular, if the same variable label occurs more than once in the clause, all its occurrences are replaced by the same graph with interface in both derivations.
Hence the head instance and the body instances produced in the reconstructed clause correspond exactly to those produced in the target clause, with predicate symbols replaced by their representatives.

The remaining body goals correspond to smaller sub-refutations in the original refutation tree. By the induction hypothesis, these body goals have refutations in $\Gamma_n$ after replacing target predicate symbols by their corresponding predicate representatives.
Therefore the whole refutation of $\leftarrow q(K)$ in $\Gamma_*$ is transformed into a refutation of $\leftarrow p_{\rho_q}(K)$ in $\Gamma_n$.
\end{proof}

\begin{lemma}\label{lem:eventual-language-inclusion}
Assume that $L_* = L(\Gamma_*,p_*) \in \FICSLFCP_{\Delta}(m,s,t,w,d)$.
Let $r_1$ be a stage after which the set $\mathcal F$ is not updated.
Then there exists a stage $r_2\ge r_1$ such that, for every $n\ge r_2$, $L(\Gamma_n,p_{\rho_\emptyset})\subseteq L_*$.
\end{lemma}

\begin{proof}
Let $\mathcal F^*$ be the fixed value of $\mathcal F$ after stage $r_1$.
Consider all bounded fixed-interface clause candidates over $\mathcal F^*\cup\{\rho_\emptyset\}$ with bounds $s,w,d$, with at most $t$ body atoms, and whose head frames belong to the fixed finite family $\mathcal H_d$, as considered in Algorithm~\ref{alg:construct-gamma}.
Since $\mathcal F^*$ is finite, $\mathcal H_d$ is finite, and the parameters $s,t,w,d$ are fixed, there are only finitely many such clause candidates.

Let $\mathcal S$ be the set of non-fact clause candidates among them that are spurious with respect to $L_*$. Thus, for each $
\gamma = p_\rho(H)\leftarrow p_{\rho_1}(H_1),\ldots,p_{\rho_\ell}(H_\ell) \in \mathcal S$ with $\ell\ge 1$, let $x_i$ be the unique variable label occurring in the star graph pattern $H_i$ for each $i=1,\ldots,\ell$, and put $X(\gamma)=\{x_i\mid 1\le i\le \ell\}$.
Since $\gamma$ is spurious, there exists a family of graphs with interfaces $(K_x^\gamma)_{x\in X(\gamma)}$ such that $\mathrm{can}(\rho_i)\odot K_{x_i}^\gamma\in L_*$ for all $i=1,\ldots,\ell$, the realization $\Real(H,\{x:=K_x^\gamma\}_{x\in X(\gamma)})$ is defined, and $\mathrm{can}(\rho)\odot \Real(H,\{x:=K_x^\gamma\}_{x\in X(\gamma)})\notin L_*$.
For each $\gamma\in\mathcal S$ and each $x\in X(\gamma)$, choose an index $i$ with $x_i=x$.
Then the graph $\mathrm{can}(\rho_i)\odot K_x^\gamma$ belongs to $L_*$.
Hence it eventually appears in the positive presentation.
When it has appeared, the graph with interface $K_x^\gamma$ occurs as a subgraph with interface of an observed positive example.
Therefore its ordered boundary representation, say $\sigma_x^\gamma$, is eventually included in $\mathrm{BRep}_w(D_n)$, and we may choose it so that $\mathrm{can}(\sigma_x^\gamma)=K_x^\gamma$.
Since $\mathcal S$ is finite, there exists $r_2\ge r_1$ such that, for every $n\ge r_2$, all such representatives $\sigma_x^\gamma$ for all $\gamma\in\mathcal S$ and all $x\in X(\gamma)$ belong to the current residual set $\mathcal R=\mathrm{BRep}_w(D_n)$.

Fix $n\ge r_2$.
We first observe that no clause candidate in $\mathcal S$ is included in $\Gamma_n$.
Indeed, let $\gamma = p_\rho(H)\leftarrow p_{\rho_1}(H_1),\ldots,p_{\rho_\ell}(H_\ell) \in\mathcal S$.
Let $x_i$ be the unique variable label occurring in the star graph pattern $H_i$, and put $X(\gamma)=\{x_i\mid 1\le i\le \ell\}$.
Take the corresponding family $(\sigma_x^\gamma)_{x\in X(\gamma)}\in\mathcal R^{X(\gamma)}$.
For each $i=1,\ldots,\ell$, we have $\mathrm{can}(\rho_i)\odot \mathrm{can}(\sigma_{x_i}^\gamma)=\mathrm{can}(\rho_i)\odot K_{x_i}^\gamma
\in L_*$.
Hence the observation table satisfies $\mathrm{Obs}_{\mathcal F,\mathcal R}(\rho_i,\sigma_{x_i}^\gamma)=\mathrm{true}$ for all $i=1,\ldots,\ell$.
Moreover, $\Real(H,\{x:=\mathrm{can}(\sigma_x^\gamma)\}_{x\in X(\gamma)})=\Real(H,\{x:=K_x^\gamma\}_{x\in X(\gamma)})$ is defined, while $\mathrm{can}(\rho)\odot \Real(H,\{x:=\mathrm{can}(\sigma_x^\gamma)\}_{x\in X(\gamma)})\notin L_*$.
Therefore the membership query in Algorithm~\ref{alg:construct-gamma} returns \textbf{false} for this family, and $\gamma$ is rejected.
Thus no spurious non-fact candidate from $\mathcal S$ is included in $\Gamma_n$.
Consequently, every non-fact clause of $\Gamma_n$ is non-spurious with respect to $L_*$. In addition, every zero-body clause $p_\rho(H)\leftarrow$ included in $\Gamma_n$ is directly sound by Algorithm~\ref{alg:construct-gamma}: it is added only if $
\mathrm{can}(\rho)\odot H$ is defined and $\mathrm{Oracle}_{L_*}(\mathrm{can}(\rho)\odot H)$ returns \textbf{true}.

We now prove the following claim by induction on the length of a refutation in $\Gamma_n$.
If $\leftarrow p_\rho(K)$ has a refutation in $\Gamma_n$, then $\mathrm{can}(\rho)\odot K\in L_*$, whenever the composition is defined.

First suppose that the refutation has length one. Then the first and only clause used is a zero-body clause $p_\rho(H)\leftarrow$ of $\Gamma_n$.
Since fixed-interface clauses have one body atom corresponding to each variable hyperedge in the head pattern, the head pattern $H$ has no variable hyperedges when the body length is zero.
Thus $H$ is ground, and the goal refuted by this clause is $\leftarrow p_\rho(H)$.
Hence $K=H$.
By the zero-body admission test in Algorithm~\ref{alg:construct-gamma}, the clause was added only when $\mathrm{can}(\rho)\odot H$ is defined and $\mathrm{Oracle}_{L_*}\bigl(\mathrm{can}(\rho)\odot H\bigr)$ returns \textbf{true}.
Therefore $\mathrm{can}(\rho)\odot H\in L_*$.
Since $K=H$, we obtain $\mathrm{can}(\rho)\odot K\in L_*$.
Thus the claim holds in the base case.

Next assume that the claim holds for all goals refutable by derivations of length less than some integer $a$, and suppose that $\leftarrow p_\rho(K)$ has a refutation of length $a$ in $\Gamma_n$.
Let the first derivation step use a non-fact clause $p_\rho(H)\leftarrow p_{\rho_1}(H_1),\ldots,p_{\rho_r}(H_r)$ of $\Gamma_n$.
For each $i=1,\ldots,r$, let $x_i$ be the unique variable label occurring in the star graph pattern $H_i$, and put $X=\{x_i\mid 1\le i\le r\}$.
Let $\theta=\{x:=K_x\}_{x\in X}$ be the substitution on variable labels used in the first derivation step, where $(K_x)_{x\in X}$ is a family of graphs with interfaces, and suppose that $\Real(H,\theta)=K$.
Then each body goal $\leftarrow p_{\rho_i}(K_{x_i})$ has a refutation of length less than $a$.
By the induction hypothesis, $\mathrm{can}(\rho_i)\odot K_{x_i}\in L_*$ for all $i=1,\ldots,r$.
Since the clause used in the first step is a non-fact clause of $\Gamma_n$, it is non-spurious with respect to $L_*$. Therefore, by Lemma~\ref{lem:one-step-preservation-by-nonspurious-clause}, $\mathrm{can}(\rho)\odot \Real(H,\{x:=K_x\}_{x\in X}) \in L_*$.
Because $\Real(H,\{x:=K_x\}_{x\in X})=K$, we obtain $\mathrm{can}(\rho)\odot K\in L_*$.
This proves the induction claim.

Finally, let $G\in L(\Gamma_n,p_{\rho_\emptyset})$.
Then $\leftarrow p_{\rho_\emptyset}(G)$ has a refutation in $\Gamma_n$.
Applying the claim with $\rho=\rho_\emptyset$, we obtain $\mathrm{can}(\rho_\emptyset)\odot G\in L_*$.
Since $\mathrm{can}(\rho_\emptyset)$ is the empty graph with the empty interface, we have $\mathrm{can}(\rho_\emptyset)\odot G=G$.
Hence $G\in L_*$.

Since $G\in L(\Gamma_n,p_{\rho_\emptyset})$ was arbitrary, we conclude $L(\Gamma_n,p_{\rho_\emptyset})\subseteq L_*$.
\end{proof}

\noindent{\bf Proof of Theorem~\ref{thm:correctness}:}\\
Let $L_* = L(\Gamma_*,p_*) \in \FICSLFCP_{\Delta}(m,s,t,w,d)$, and let $G_1,G_2,\ldots$ be an arbitrary positive presentation of $L_*$.
By Lemma~\ref{lem:eventual-target-clause-reconstruction}, there exists an integer $n_1$ such that, for every update stage $n'\ge n_1$, the hypothesis $(\Gamma_{n'},p_{\rho_\emptyset})$ contains predicate representatives for all target predicate contexts, and every target clause of $\Gamma_*$ is reconstructed in $\Gamma_{n'}$ over those representatives.

We distinguish two cases.

First suppose that there exists an update stage $n_2\ge n_1$.
Then, by Lemma~\ref{lem:eventual-target-clause-reconstruction}, every target clause of $\Gamma_*$ is reconstructed in $\Gamma_{n_2}$.
Hence, by Lemma~\ref{lem:refutation-simulation}, $L_*\subseteq L(\Gamma_{n_2},p_{\rho_\emptyset})$.
Therefore $D_{n_2+1}\subseteq L_*\subseteq L(\Gamma_{n_2},p_{\rho_\emptyset})$.
Since the hypothesis used in the update test at stage $n_2+1$ is precisely the hypothesis constructed from the current pair $(\mathcal F,\mathcal R)$ at stage $n_2$, the update condition in Algorithm~\ref{alg:learning} fails at stage $n_2+1$.
Thus $\mathcal F$ is not updated at stage $n_2+1$.
Moreover, after stage $n_2$, the predicate basis $\mathcal F$ still contains the representatives of all target predicate contexts.
Although the residual set $\mathcal R$ is updated to $\BRep_w(D_n)$ at each stage, the reconstructed target clauses remain admissible for every later $\mathcal R$.
Indeed, the argument in the proof of Lemma~\ref{lem:eventual-target-clause-reconstruction} shows that genuine target clauses pass the membership-query test in Algorithm~\ref{alg:construct-gamma} for every family of residual representatives for which the body observations are positive.
Consequently, the same simulation argument as in Lemma~\ref{lem:refutation-simulation} gives $L_*\subseteq L(\Gamma_{n_2+1},p_{\rho_\emptyset})$.
It follows that $D_{n_2+2}\subseteq L_*\subseteq L(\Gamma_{n_2+1},p_{\rho_\emptyset})$, and hence no update is performed at stage $n_2+2$.
Repeating this argument inductively, no update is performed at any stage after $n_2$, and for every $n\ge n_2$, $L_*\subseteq L(\Gamma_n,p_{\rho_\emptyset})$.
In particular, $n_2$ is a stage after which the set $\mathcal F$ is not updated.
Now apply Lemma~\ref{lem:eventual-language-inclusion} to this stage $n_2$.
There exists a stage $n_3\ge n_2$ such that, for every $n\ge n_3$, $L(\Gamma_n,p_{\rho_\emptyset})\subseteq L_*$.

Combining the two inclusions, we obtain $L(\Gamma_n,p_{\rho_\emptyset})=L_*$ for every $n\ge n_3$.

It remains to consider the second case, where there is no update stage $n'\ge n_1$.
Then $n_1$ is a stage after which the set $\mathcal F$ is not updated.
Since $\mathcal F$ is fixed, the structural parameters $s,t,w,d$ are fixed, and the head frames are restricted to the fixed finite family $\mathcal H_d$, there are only finitely many bounded fixed-interface clause candidates over $\mathcal F\cup\{\rho_\emptyset\}$.
As $n$ increases, the residual set $\mathcal R=\BRep_w(D_n)$ can only grow.
Therefore the clause-admission test in Algorithm~\ref{alg:construct-gamma} can only reject additional non-fact candidates; it cannot make a previously rejected candidate admissible again.
The zero-body clauses are independent of $\mathcal R$.
Hence the sequence of clause sets $\Gamma_n$, for $n\ge n_1$, is descending inside a finite set of candidates.
Therefore there exists a stage $n_4\ge n_1$ such that $\Gamma_n=\Gamma_{n_4}$ for every $n\ge n_4$.
We show that $L_*\subseteq L(\Gamma_{n_4},p_{\rho_\emptyset})$.
Let $G\in L_*$. Since $G_1,G_2,\ldots$ is a positive presentation of $L_*$, there exists $j\ge n_4+1$ such that $G\in D_j$.
By the assumption of the present case, no update is performed at stage $j$.
Hence the update condition in Algorithm~\ref{alg:learning} fails at stage $j$, and therefore $D_j\subseteq L(\Gamma_{j-1},p_{\rho_\emptyset})$.
Since $j-1\ge n_4$, we have $\Gamma_{j-1}=\Gamma_{n_4}$.
Thus $G\in L(\Gamma_{n_4},p_{\rho_\emptyset})$.
Because $G\in L_*$ was arbitrary, $L_*\subseteq L(\Gamma_{n_4},p_{\rho_\emptyset})$.
Since $\Gamma_n=\Gamma_{n_4}$ for all $n\ge n_4$, it follows that $L_*\subseteq L(\Gamma_n,p_{\rho_\emptyset})$ for every $n\ge n_4$.
On the other hand, since $n_1$ is a stage after which the set $\mathcal F$ is not updated, Lemma~\ref{lem:eventual-language-inclusion} gives a stage $n_5\ge n_1$ such that, for every $n\ge n_5$, $L(\Gamma_n,p_{\rho_\emptyset})\subseteq L_*$.
Taking $n_0=\max\{n_4,n_5\}$, we obtain $L(\Gamma_n,p_{\rho_\emptyset})=L_*$ for every $n\ge n_0$.
In the first case, put $n_0=n_3$. In the second case, take $n_0=\max\{n_4,n_5\}$ as above.
Hence, in both cases, there exists an integer $n_0$ such that $L(\Gamma_n,p_{\rho_\emptyset})=L_*$ for all $n\ge n_0$.

Therefore Algorithm~\ref{alg:learning} identifies $L_*$ in the limit from positive data and membership queries.
\hfill$\Box$
\subsection{Polynomial-time update}\label{subsec:complexity}

We next analyze the update complexity of Algorithm~\ref{alg:learning}.
The main point is that, under the bounded-degree assumption, the fixed interface-rank bound, and the finite head-frame condition, only polynomially many ordered boundary representations and clause candidates need to be considered at each stage.

The analysis proceeds as follows. We first bound the size and construction time of $\mathrm{BRep}_w(D_n)$.
This immediately bounds the number of predicate symbols introduced from the current predicate basis.
We then bound the number of bounded fixed-interface clause candidates and the number of membership queries used in the observation table and clause-admission tests.
Finally, these estimates are combined with the membership bound of Proposition~\ref{prop:membership-bound} for testing membership in the current hypothesis, while membership queries to the target language are counted as oracle calls.

\begin{lemma}\label{lem:brep-size-bound}
Let $D_n=\{G_1,\ldots,G_n\}$ and let $S_n=\sum_{i=1}^{n}|V(G_i)|$.
Assume that every $G_i$ has maximum degree at most $\Delta$.
For fixed $w$ and $\Delta$, the size of $\mathrm{BRep}_w(D_n)$ is bounded by a polynomial in $S_n$. 
More precisely, there exists a function $\xi_1$, depending only on $w$ and $\Delta$, such that $|\mathrm{BRep}_w(D_n)|\leq\xi_1(w,\Delta) S_n^w$.
Moreover, $\mathrm{BRep}_w(D_n)$ can be constructed in time $O\!\left(\Delta\,2^{w\Delta} S_n^{w+1}\right)$.
\end{lemma}

\begin{proof}
Fix $i\in\{1,\ldots,n\}$, and put $N_i=|V(G_i)|$.
For a fixed boundary rank $r\leq w$, the number of ordered boundary tuples $\beta=(b_1,\ldots,b_r)$ of distinct vertices of $G_i$ is at most $N_i^r$.
For each such tuple, since $G_i$ has maximum degree at most $\Delta$, the number of edges incident with the boundary vertices is at most $r\Delta$.
Hence the number of possible boundary edge sets $E_B$ is at most $2^{r\Delta}$.
Therefore the number of ordered boundary specifications of rank $r$ arising from $G_i$ is at most $N_i^r 2^{r\Delta}$.
Summing over $r=0,1,\ldots,w$, the number of ordered boundary representations of rank at most $w$ arising from $G_i$ is at most $\sum_{r=0}^{w} N_i^r 2^{r\Delta}$.
Thus
\[
|\mathrm{BRep}_w(D_n)|
\leq
\sum_{i=1}^{n}\sum_{r=0}^{w} N_i^r 2^{r\Delta}.
\]
Since $w$ and $\Delta$ are fixed, there is a function $\xi_1$ depending only on $w$ and $\Delta$ such that $\sum_{r=0}^{w} N_i^r 2^{r\Delta}\leq\xi_1(w,\Delta)N_i^w$for every $i$.
Hence
\[
|\mathrm{BRep}_w(D_n)|
\leq
\xi_1(w,\Delta)\sum_{i=1}^{n} N_i^w.
\]
Since $\sum_{i=1}^{n} N_i^w\leq\left(\sum_{i=1}^{n} N_i\right)^w=S_n^w$, we obtain
\[
|\mathrm{BRep}_w(D_n)|
\leq
\xi_1(w,\Delta)S_n^w.
\]

The construction time bound follows from Algorithm~\ref{alg:construct-obrep}.
For each candidate $(\beta,E_B)$, validity can be tested and the corresponding subgraph with interface can be constructed in $O(|V(G_i)|+|E(G_i)|)$ time.
Since $G_i$ has maximum degree at most $\Delta$, this is $O(\Delta N_i)$.
Therefore the total construction time is bounded by
\[
\sum_{i=1}^{n}
O\!\left(
\sum_{r=0}^{w} N_i^r2^{r\Delta}\cdot \Delta N_i
\right)
=
O\!\left(
\Delta\,2^{w\Delta}\sum_{i=1}^{n} N_i^{w+1}
\right).
\]
Since $\sum_{i=1}^{n} N_i^{w+1}\leq\left(\sum_{i=1}^{n}N_i\right)^{w+1}=S_n^{w+1}$, we obtain
\[
O\!\left(\Delta\,2^{w\Delta}S_n^{w+1}\right).
\]
This proves the lemma.
\end{proof}

\begin{lemma}\label{lem:polynomially-many-predicate-symbols}
Let $D_n=\{G_1,\ldots,G_n\}$ and let $S_n=\sum_{i=1}^{n}|V(G_i)|$.
Assume that every $G_i$ has maximum degree at most $\Delta$.
For fixed $w$ and $\Delta$, the number of predicate symbols introduced by Algorithm~\ref{alg:construct-gamma} at stage $n$ is bounded by a polynomial in $S_n$.
\end{lemma}

\begin{proof}
Algorithm~\ref{alg:construct-gamma} introduces one predicate symbol $p_\rho$ for each $\rho\in \mathcal F\cup\{\rho_\emptyset\}$.
At stage $n$, the set $\mathcal F$ is either unchanged from an earlier stage or is set to $\mathrm{BRep}_w(D_j)$ for some $j\le n$.
Since $D_j\subseteq D_n$, we have $\mathrm{BRep}_w(D_j)\subseteq \mathrm{BRep}_w(D_n)$.
Hence $|\mathcal F|\le |\mathrm{BRep}_w(D_n)|$.
Therefore the number of predicate symbols introduced is at most $|\mathrm{BRep}_w(D_n)|+1$.
By Lemma~\ref{lem:brep-size-bound}, this quantity is bounded by a polynomial in $S_n$.
\end{proof}

\begin{lemma}\label{lem:polynomially-many-clause-candidates}
Let $D_n=\{G_1,\ldots,G_n\}$ and let $S_n=\sum_{i=1}^{n}|V(G_i)|$.
Assume that every $G_i$ has maximum degree at most $\Delta$.
For fixed $s,t,w,d,\Delta$, the number of bounded fixed-interface clause candidates considered by Algorithm~\ref{alg:construct-gamma} at stage $n$ is bounded by a polynomial in $S_n$.
This count includes the zero-body case $\ell=0$.
\end{lemma}

\begin{proof}
Let $M_n=|\mathrm{BRep}_w(D_n)|$.
By Lemma~\ref{lem:brep-size-bound}, $M_n$ is bounded by a polynomial in $S_n$.
At stage $n$, the current predicate basis $\mathcal F$ is either unchanged from an earlier stage or is equal to $\mathrm{BRep}_w(D_j)$ for some $j\le n$.
Hence $|\mathcal F|\le M_n$.
Thus the number of available predicate representatives is at most $|\mathcal F|+1\le M_n+1$, where the additional one is $\rho_\emptyset$.

A clause candidate considered by Algorithm~\ref{alg:construct-gamma} has the form $p_\rho(H)\leftarrow p_{\rho_1}(H_1),\ldots,p_{\rho_\ell}(H_\ell)$, where $0\le \ell\le t$, and all predicate representatives are chosen from $\mathcal F\cup\{\rho_\emptyset\}$.
For a fixed $\ell$, the number of choices of the head and body predicate representatives is at most $(M_n+1)^{\ell+1}\le (M_n+1)^{t+1}$.
By Algorithm~\ref{alg:construct-gamma} and Definition~\ref{def:bounded-fics}, the head frame of $H$ is required to belong to $\mathcal H_d$.
For fixed parameters $s,t,w,d$ and finite alphabets $\Sigma_V,\Sigma_E$, the number of possible head-frame types is bounded by $|\mathcal H_d|$, which is a constant depending only on $d$ and the fixed label alphabets.
For each such head frame, the number of possible assignments of variable labels to at most $s$ variable-hyperedge occurrences, up to equality pattern and subject to the rank bound $w$, is bounded by a constant depending only on $s$ and $w$.
The body star patterns are then determined up to finitely many choices depending only on the corresponding variable ranks and the finite label alphabets.

Therefore the total number of bounded fixed-interface clause candidates, including the case $\ell=0$, is bounded by $c\sum_{\ell=0}^{t}(M_n+1)^{\ell+1}$ for some constant $c$ depending only on $s,t,w,d,\Sigma_V,\Sigma_E$.
This is bounded by a polynomial in $M_n$, and hence by a polynomial in $S_n$.
\end{proof}

\begin{lemma}\label{lem:polynomially-many-membership-queries-per-update}
Let $D_n=\{G_1,\ldots,G_n\}$ and let $S_n=\sum_{i=1}^{n}|V(G_i)|$.
Assume that every $G_i$ has maximum degree at most $\Delta$.
For fixed $s,t,w,d,\Delta$, the total number of membership queries used by Algorithm~\ref{alg:construct-gamma} at stage $n$ is bounded by a polynomial in $S_n$.
\end{lemma}

\begin{proof}
Let $M_n=|\mathrm{BRep}_w(D_n)|$.
By Lemma~\ref{lem:brep-size-bound}, $M_n$ is bounded by a polynomial in $S_n$.
At stage $n$, the current predicate basis $\mathcal F$ is either unchanged from an earlier stage or is equal to $\mathrm{BRep}_w(D_j)$ for some $j\le n$.
Hence $|\mathcal F|\le M_n$.
Moreover, Algorithm~\ref{alg:learning} sets $\mathcal R=\mathrm{BRep}_w(D_n)$ at stage $n$, and therefore $|\mathcal R|=M_n$.

First, consider the construction of the observation table $\mathrm{Obs}_{\mathcal F,\mathcal R}$.
For each pair $(\rho,\sigma)\in \mathcal F\times\mathcal R$, Algorithm~\ref{alg:construct-observation-table} issues at most one
membership query, namely for $\mathrm{can}(\rho)\odot \mathrm{can}(\sigma)\in L_*$, provided that the composition is defined. Hence the number of membership queries used for the observation table is at most $|\mathcal F|\,|\mathcal R|\le M_n^2$.

Next, consider the clause-admission tests in Algorithm~\ref{alg:construct-gamma}.
By Lemma~\ref{lem:polynomially-many-clause-candidates}, the number of bounded fixed-interface clause candidates considered at stage $n$, including the zero-body case $\ell=0$, is bounded by a polynomial in $S_n$.
For a zero-body candidate $p_\rho(H)\leftarrow$, Algorithm~\ref{alg:construct-gamma} issues at most one membership query, namely for $\mathrm{can}(\rho)\odot H\in L_*$, provided that the composition is defined.
Therefore the total number of membership queries used for zero-body candidates is bounded by the number of clause candidates, and hence by a polynomial in $S_n$.
Now consider a non-fact candidate $\gamma = p_\rho(H)\leftarrow p_{\rho_1}(H_1),\ldots,p_{\rho_\ell}(H_\ell)$, where $1\le \ell\le t$.
For each $i=1,\ldots,\ell$, let $x_i$ be the unique variable label occurring in the star graph pattern $H_i$, and put $X(\gamma)=\{x_i\mid 1\le i\le \ell\}$.
Algorithm~\ref{alg:construct-gamma} checks families $(\sigma_x)_{x\in X(\gamma)}\in\mathcal R^{X(\gamma)}$.
Since $|\mathcal R|=M_n$, the number of such families is $M_n^{|X(\gamma)|}$.
Since $|X(\gamma)|\le \ell\le t$, this number is at most $M_n^t$.

For each such family, Algorithm~\ref{alg:construct-gamma} issues at most one head membership query, namely for $\mathrm{can}(\rho)\odot \Real(H,\{x:=\mathrm{can}(\sigma_x)\}_{x\in X(\gamma)})\in L_*$, whenever the body observations are positive and the realization is defined.
Hence the number of membership queries used for a single non-fact candidate is at most $M_n^t$.
Multiplying this bound by the polynomially many clause candidates from Lemma~\ref{lem:polynomially-many-clause-candidates}, we obtain a polynomial bound in $S_n$ for all non-fact clause-admission queries.

Combining the observation-table queries, the zero-body clause queries, and the non-fact clause-admission queries, the total number of membership queries used by Algorithm~\ref{alg:construct-gamma} at stage $n$ is bounded by a polynomial in $S_n$.
\end{proof}

\begin{lemma}\label{lem:polynomial-time-one-update}
At each stage $n$, the total time required to construct the hypothesis $(\Gamma_n,p_{\rho_\emptyset})$ from $G_1,\ldots,G_n$ is bounded by a polynomial in $S_n=\sum_{i=1}^{n}|V(G_i)|$.
\end{lemma}

\begin{proof}
We estimate the running time of stage $n$ of Algorithm~\ref{alg:learning}.
At this stage, the algorithm constructs the hypothesis $(\Gamma_n,p_{\rho_\emptyset})$ from the positive examples $G_1,\ldots,G_n$.

First consider the update test of Algorithm~\ref{alg:learning}.
It tests whether $D_n\nsubseteq L(\Gamma',p_{\rho_\emptyset})$, where $\Gamma'$ is the clause system $\Gamma(\mathcal F,\mathcal R)$ constructed by Algorithm~\ref{alg:construct-gamma}.
The membership tests used in the update condition are tests for the current hypothesis $(\Gamma',p_{\rho_\emptyset})$ and are ordinary computations, not membership queries to the target language.
For this test, we check whether $G_i\in L(\Gamma',p_{\rho_\emptyset})$ for each graph $G_i\in D_n$.
For each such graph, the membership problem for $(\Gamma',p_{\rho_\emptyset})$ is solvable in polynomial time by Proposition~\ref{prop:membership-bound} and Corollary~\ref{cor:membership-bound}, since the structural parameters are fixed.
Since $|D_n|=n\le S_n$, the total cost of the update test is bounded by a polynomial in $S_n$.

Next, if the update test succeeds, Algorithm~\ref{alg:learning} constructs $\mathrm{BRep}_w(D_n)$ and assigns it to $\mathcal F$.
Regardless of whether $\mathcal F$ is updated, the algorithm also constructs $\mathcal R=\mathrm{BRep}_w(D_n)$ at stage $n$.
By Lemma~\ref{lem:brep-size-bound}, $\mathrm{BRep}_w(D_n)$ can be constructed in time polynomial in $S_n$.
Moreover, its size is also bounded by a polynomial in $S_n$.

It remains to estimate the cost of constructing $\Gamma_n=\Gamma(\mathcal F,\mathcal R)$ by Algorithm~\ref{alg:construct-gamma}.
By Lemma~\ref{lem:polynomially-many-predicate-symbols}, the number of predicate symbols introduced by Algorithm~\ref{alg:construct-gamma} is bounded by a polynomial in $S_n$.
By Lemma~\ref{lem:polynomially-many-clause-candidates}, the number of bounded fixed-interface clause candidates considered by Algorithm~\ref{alg:construct-gamma}, including the zero-body case $\ell=0$, is bounded by a polynomial in $S_n$.
The construction of the observation table and the clause-admission tests use membership queries to $L_*$.
By Lemma~\ref{lem:polynomially-many-membership-queries-per-update}, the total number of membership queries used by Algorithm~\ref{alg:construct-gamma} at stage $n$ is bounded by a polynomial in $S_n$.
For each such membership query, the queried graph is either of the form $\mathrm{can}(\rho)\odot \mathrm{can}(\sigma)$ for the observation table, or of the form $\mathrm{can}(\rho)\odot H$ for a zero-body clause candidate, or of the form $\mathrm{can}(\rho)\odot \Real(H,\{x:=\mathrm{can}(\sigma_x)\}_{x\in X(\gamma)})$ for a non-fact clause candidate $\gamma$.
Since the structural parameters are fixed and all representatives used at stage $n$ come from $\mathrm{BRep}_w(D_n)$, the size of each queried graph is bounded by a polynomial in $S_n$.
By the convention on membership queries, the answers to these queries are provided by the oracle. Lemma~\ref{lem:polynomially-many-membership-queries-per-update} shows that the number of queries is polynomial in $S_n$, and the queried graphs constructed above have size polynomially bounded in $S_n$.

All remaining internal operations in Algorithm~\ref{alg:construct-gamma} and Algorithm~\ref{alg:learning}, such as checking whether compositions and realizations are defined and constructing the corresponding graphs with interfaces, are performed over polynomially many candidates and graphs of polynomial size.
Hence their total cost is polynomial in $S_n$.

Combining the above bounds, every component of the learner's internal construction of $(\Gamma_n,p_{\rho_\emptyset})$ at stage $n$ is bounded by a polynomial in $S_n$.
Therefore the internal time required to compute $(\Gamma_n,p_{\rho_\emptyset})$ at stage $n$ is bounded by a polynomial in $S_n$.
\end{proof}

\begin{theorem}[Polynomial-time update]\label{thm:polynomial-time-update}
Algorithm~\ref{alg:learning} has polynomial-time update in the sense of Definition~\ref{def:poly-update}.
\end{theorem}

\begin{proof}
Let $L_* \in \FICSLFCP_{\Delta}(m,s,t,w,d)$, and let $G_1,G_2,\ldots$ be a positive presentation of $L_*$.
For each stage $n$, Lemma~\ref{lem:polynomial-time-one-update} shows that the total time required to construct the $n$-th hypothesis $(\Gamma_n,p_{\rho_\emptyset})$ from the initial segment $G_1,\ldots,G_n$ is bounded by a polynomial in $S_n=\sum_{i=1}^{n}|V(G_i)|$.
Therefore there exists a polynomial $P_{\mathrm{upd}}$ such that, for every $n\ge 1$, the time required by Algorithm~\ref{alg:learning} to output $(\Gamma_n,p_{\rho_\emptyset})$ after reading $G_1,\ldots,G_n$ is at most $P_{\mathrm{upd}}(S_n)$.
Hence Algorithm~\ref{alg:learning} has polynomial-time update in the sense of Definition~\ref{def:poly-update}.
\end{proof}

The preceding bounds show that the learner stores only polynomially many ordered boundary representations in $\mathrm{BRep}_w(D_n)$, generates only polynomially many predicate symbols and bounded fixed-interface clause candidates, and performs only polynomially many membership queries at each stage.
Combined with the parameterized polynomial bound for the membership problem, this yields polynomial-time update for Algorithm~\ref{alg:learning}.

\section{Conclusions}\label{sec:conclusions}

We have studied the learnability of graph languages generated by fixed-interface clause systems from the viewpoint of distributional learning.
This paper is a full journal version of our earlier conference paper~\cite{shoudai2016ilp}, in which all proofs were omitted. 
It extends that work in four respects: complete proofs are provided, the framework is reformulated using fixed-interface clause systems and ordered boundary representations that make the boundary structure explicit, the bounded treewidth assumption of~\cite{shoudai2016ilp} is shown to be unnecessary for the present learning argument, and an explicit parameter tuple is introduced to trace the structural restrictions responsible for learnability and polynomial-time update.

A key point of the present approach is that, although the target language consists of plain graphs, the learner operates on graphs with interfaces obtained by cutting observed graphs along designated boundary vertices.
This explicit boundary structure is what makes it possible to define context candidates in a graph setting and to carry out a parameterized complexity analysis of the update process.
In this sense, the present formulation makes explicit a structural ingredient that is only implicit in fragment-based graph learning formalisms.

Within this framework, we considered a bounded class of graph languages satisfying the finite context property under a bounded-degree assumption.
The bounds were made explicit through the degree bound $\Delta$ and the structural parameter tuple $(m,s,t,w,d)$, which together control the generated graph class, the shape of the target clause systems, and the size of the candidate hypothesis space.
In particular, the finite head-frame condition should be viewed as the fixed-interface counterpart of a bounded rule-shape assumption.
We presented an oracle-guided learning algorithm based on ordered boundary representations induced by the observed positive examples and on membership-query-based clause admission.
The correctness proof was organized into separate steps: target contexts eventually appear in the observed sample, target clauses are reconstructed over the corresponding predicate representatives, and spurious non-fact clauses are eventually excluded by membership queries.
From these ingredients, we obtained identification in the limit from positive data and membership queries, together with polynomial-time update, for the class $\FICSLFCP_{\Delta}(m,s,t,w,d)$.

Several directions remain open.
First, the present analysis is restricted to unary predicate symbols.
It remains to investigate whether analogous parameterized bounds can be obtained for clause systems with higher-arity predicates.
Second, it would be desirable to make more explicit the formal relationship between fixed-interface clause systems and regular formal graph systems~\cite{uchida1995parallel,shoudai2016ilp}, and to clarify in a precise form how the current formulation extends the regular-FGS generation mechanism.
Third, the interface-degree-safety condition described in Remark~\ref{rem:interface-degree-safe} gives a sufficient but not necessary syntactic condition for bounded degree.
Identifying weaker conditions that still guarantee bounded degree under recursive derivations is an important direction. Fourth, and most intriguingly, the present work suggests a connection to factor-critical graphs.
It remains open whether factor-critical graphs can be generated by fixed-interface clause systems, and in particular whether they can be generated in the more restricted regular-FGS setting.
Although more expressive non-regular FGS variants can express such graphs, it is not known whether the fixed-interface setting, or even the more restricted regular-FGS setting, suffices.
Identifying the precise boundary of factor-critical graph classes expressible within $\FICSL_{\Delta}(m,s,t,w,d)$ is a challenging and promising direction for future work.

\section*{Acknowledgments}
This work was partially supported by JSPS KAKENHI Grant Numbers JP24K15074 and JP24K15090.

\bibliographystyle{fundam}
\bibliography{main}

\end{document}